\documentclass{llncs}
\usepackage{amsmath,amssymb,amsfonts}
\usepackage{array}
\usepackage{float}
\usepackage{enumerate}
\usepackage{ifthen}
\usepackage{xspace}
\usepackage{graphicx}
\usepackage{arydshln}
\usepackage{subfigure}

\usepackage{algorithmic}
\usepackage{algorithm}

\floatstyle{ruled}
\newfloat{algorithm}{thp}{lop}[section]
\floatname{algorithm}{Algorithm}
\newfloat{module}{thp}{lop}[section]
\floatname{module}{Module}

\newtheorem{thm}{Theorem}[section]
\newtheorem{cor}[thm]{Corollary}
\newtheorem{lem}[thm]{Lemma}

\newcommand{\remove}[1]{}

\begin{document}

%numérotation des pages
%\frontmatter
%\pagenumbering{arabic}
%\pagestyle{plain}

\title{RoboCast: Asynchronous Communication in Robot Networks} 

\author{Zohir Bouzid$^\star$ \and Shlomi Dolev$^\dag$ \and Maria Potop-Butucaru$^\star$ \and S\'{e}bastien Tixeuil$^\star$}

\author{Zohir Bouzid\inst{1}\thanks{Supported by DIGITEO project PACTOLE and the ANR projects R-DISCOVER and SHAMAN.}
\and Shlomi Dolev\inst{2}\thanks{Part of the research was done during a
supported visit of Shlomi Dolev at LIP6 Université Pierre et Marie
Curie - Paris 6. Partially supported by Rita Altura trust chair in
computer sciences, ICT Programme of the European Union under contract
number FP7-215270 (FRONTS), and US Air Force European Office of
Aerospace Research and Development, grant number FA8655-09-1-3016.
}
\and Maria Potop-Butucaru\inst{1}
\and S\'{e}bastien Tixeuil\inst{1}
}

\institute{Universit\'e Pierre et Marie Curie - Paris 6, France\\ 
\and
Ben Gurion University of the Negev, Israel}

%\date{\normalsize $^\star$ Universit\'e Pierre et Marie Curie - Paris 6, France\\
%\normalsize $^{\star\star}$ Ben Gurion University of the Negev, Israel}

\remove{

\institute{$^\star$ Universit\'e Pierre et Marie Curie - Paris 6, France\\
$^\dag$ Ben Gurion University of the Negev, Israel}
%Universit\'e Pierre et Marie Curie - Paris 6, LIP6-CNRS 7606, France}

\date{\normalsize $^\star$ Universit\'e Pierre et Marie Curie - Paris 6, France\\
\normalsize $^\dag$ Ben Gurion University of the Negev, Israel}

}

\maketitle

\begin{abstract}
This paper introduces the \emph{RoboCast} communication abstraction. The RoboCast allows a swarm of non oblivious, anonymous robots that are only endowed with visibility sensors and do not share a common coordinate system, to asynchronously exchange information. 
We propose a generic framework that covers a large class of asynchronous communication algorithms and show how our framework can be used to implement fundamental building blocks in robot networks such as gathering or stigmergy. In more details, we propose a RoboCast algorithm that allows robots to broadcast their local coordinate systems to each others. Our algorithm is further refined with a local collision avoidance scheme. Then, using the RoboCast primitive, we propose algorithms for deterministic asynchronous gathering and binary information exchange.

\remove{
~\\~\\~\\~\\~\\~\\~\\~\\~\\~\\~\\~\\~\\~\\

\noindent
\textbf{Regular submission}\\~\\~\\
\textbf{Contact author}: \\
Zohir BOUZID \\
Email: \emph{zohir.bouzid@lip6.fr} \\
Address: LIP6, 4, place Jussieu, 75015 Paris, France.\\
Phone: +33 1 44 27 88 77\\~\\
%Please consider this \textbf{regular} submission for the best student paper award, Zohir BOUZID is a PhD student.
}

\end{abstract}

%\setcounter{page}{0}
%\thispagestyle{empty}
%\newpage

\section{Introduction}
Existing studies in robots networks focus on characterizing the computational power of these systems when robots are endowed with visibility sensors and communicate using \emph{only} their movements without relying on any sort of agreement on a global coordinate system. Most of these studies \cite{ando1999dmp,DK02,CFPS03} assume oblivious robots (\emph{i.e.} robots have no persistent memory of their past actions), so the ``memory'' of the network is implicit and generally deduced from the current positions of the robots. Two computation models are commonly used in robot networks: ATOM \cite{SY99} and CORDA \cite{FPSW99h}. In both models robots perform in Look-Compute-Move cycles. The main difference is that these cycles are executed in a fully asynchronous manner in the CORDA model while each phase of the Look-Compute-Move cycle is executed in a lock step fashion in the ATOM model. These computation models have already proved their limitations. That is, the deterministic implementations of many fundamental abstractions such as gathering or leader election are proved impossible in these settings without additional assumptions (\cite{Pre05,CP07}). The purpose of this paper is to study how the addition of \emph{bounded} memory to each individual robot can increase the computational power of an \emph{asynchronous} swarm of robots. 
%We focus on the \emph{broadcast} communication primitive, 
We focus on an \emph{all-to-all} communication primitive,
called RoboCast, which is a basic building block for the design of any distributed system. A positive answer to this problem is the open gate for solving fundamental problems for robot networks such as gathering, scattering, election or exploration. 

In robot networks, using motion to transmit information is not new \cite{SY99,SY99b,DDPS09}. In~\cite{SY99}, Suzuki and Yamashita present an algorithm for broadcasting the local coordinate system of each robot (and thus build a common coordinate system) under the ATOM model. The algorithm heavily relies on the phase atomicity in each Look-Compute-Move cycle. In particular, a robot $a$ that observes another robot $b$ in four distinct positions has the certitude that $b$ has in turn already seen $a$ in at least two different positions. The situation becomes more intricate in the asynchronous CORDA model. Indeed, the number of different positions observed for a given robot is not an indicator on the number of complete cycles executed by that robot since cycles are completely uncorrelated. By contrast, our implementation of RoboCast is designed for the more general CORDA model and uses a novel strategy: the focus moves from observing robots in different positions to observing robots moving in different \emph{directions}. That is, each robot changes its direction of movement when a particular stage of the algorithm is completed; this change allows the other robots to infer information about the observed robot.

Another non trivial issue that needs to be taken care of without explicit communication is \emph{collisions avoidance}, since colliding robots could be confused due to indistinguishability. Moreover, robots may physically collide during their Move phase. One of the techniques commonly used to avoid collisions consists in computing a Voronoi diagram~\cite{A91} and allowing robots to move \emph{only} inside their Voronoi cells~\cite{DK02}. Since the Voronoi cells do not overlap with one another, robots are guaranteed to not collide. This simple technique works well in the ATOM model but heavily relies on the computation of the same Voronoi diagram by the robots that are activated concurrently, and thus does not extend to the CORDA model where different Voronoi diagrams 
may be computed by different robots, inducing possible collisions. Our approach defines a collision-free zone of movement that is compatible with the CORDA model constraints.

Applications of our RoboCast communication primitive include fundamental services in robot networks such as gathering and stigmergy. 
Deterministic gathering of two stateless robots has already been proved impossible when robots have no common orientation~\cite{SY99}.
In~\cite{SY99}, the authors also propose non-oblivious solutions for deterministic gathering in the ATOM model. Our RoboCast permits to extend this result to the CORDA model, using bounded memory and a limited number of movements. Recently, in \cite{DDPS09}, the authors extend the work of \cite{SY99} to efficiently implement stigmergy in robot networks in the ATOM model. Stigmergy is the ability for robots to exchange binary information that is encoded in the way they move. This scheme is particularly appealing for secure communication in robot networks, since \emph{e.g.} jamming has no impact on robot communication capability. The RoboCast primitive allows to extend this mechanism to the CORDA model, with a collision-free stigmergy scheme.

\paragraph{\textbf{Our contribution}}
We formally specify a robot network communication primitive, called RoboCast, and propose implementation variants for this primitive, that permit anonymous robots not agreeing on a common coordinate system, to exchange various information (\emph{e.g.} their local coordinate axes, unity of measure, \emph{rendez-vous} points, or binary information) using only motion in a two dimensional space. Contrary to previous solutions, our protocols all perform in the \emph{fully asynchronous} CORDA model, use \emph{constant} memory and a \emph{bounded} number of movements.
Then, we use the RoboCast primitive to efficiently solve some fundamental open problems in robot networks. We present a fully asynchronous deterministic gathering and a fully asynchronous stimergic communication scheme. 
Our algorithms differ from previous works by several key features: they are totally asynchronous (in particular they do not rely on the atomicity of cycles executed by robots), they make no assumption on a common chirality or knowledge of the initial positions of robots, and finally, each algorithm uses only a bounded number of movements. Also, for the first time in these settings, our protocols use CORDA-compliant collision avoidance schemes.

\paragraph{\textbf{Roadmap}}
The paper is made up of six 
sections. Section~\ref{sec:model} describes the computing
model and presents the formal specification of
the RoboCast problem.
Section~\ref{sec:lcs}
presents our protocol and its complexity.
The algorithm is enhanced in Section~\ref{sec:collisionfree} with a collision-avoidance scheme.
Using the Robocast primitive, Section~\ref{sec:applications} proposes
algorithms for deterministic asynchronous gathering and binary information exchange.
Finally, Section~\ref{sec:conclusion} provides concluding
remarks. 
Some proofs are relegated to the appendix.
%Due to space limitations, some proofs are relegated in an appendix. 

\section{Model}
\label{sec:model}

We consider a network that consists of a finite set of $n$ robots arbitrarily deployed in a two dimensional space, with no two robots located at the same position. 
Robots are devices with sensing, computing and moving capabilities. They can observe (sense) the positions of other robots in the space and based on these observations, they perform some local computations that can drive them to other locations. 

In the context of this paper, the robots are \emph{anonymous}, in the sense that they can not be distinguished using their appearance and they do not have any kind of identifiers that can be used during the computation. In addition, there is no direct mean of communication between them. Hence, the only way for robots to acquire information is by observing their positions. 
Robots have \emph{unlimited visibility}, \emph{i.e.} they are able to sense the entire set of robots. 
%Robots are also equipped with a strong multiplicity sensor that provides robots with the ability to detect the exact number of robots that may simultaneously occupy the same location. 
We assume that robots are \emph{non-oblivious}, 
\emph{i.e.} they can remember observations, computations and motions performed in previous steps.
%We assume that the robots remember observation and computations performed in previous steps. That is, robots have a local uncorruptible memory. 
Each robot is endowed with a local coordinate system and a local unit measure which may be different from those of other robots.
This local coordinate system is assumed to be fixed during a run unless it is explicitly modified by the corresponding robot as a result of a computation.
We say in this case that robots \emph{remember} their own coordinate systems. 
This is a common assumption when studying non-oblivious robot networks \cite{SY99,DDPS09}.

A \emph{protocol} is a collection of $n$ \emph{programs}, one operating on each robot. The program of a robot consists in executing {\em Look\mbox{-}Compute\mbox{-}Move cycles} infinitely many times. That is, the robot first observes its environment (Look phase). An observation returns a snapshot of the positions of all robots within the visibility range. In our case, this observation returns a snapshot of the positions of \emph{all} robots. The observed positions are \emph{relative} to the observing robot, that is, they use the coordinate system of the observing robot. Based on its observation, a robot then decides --- according to its program --- to move or to stay idle (Compute phase). When a robot decides a move, it moves to its destination during the Move phase.  

The local state of a robot is defined by the content of its memory and its position. A configuration of the system is the union of the local states of all the robots in the system. An \emph{execution} $e=(c_0, \ldots, c_t, \ldots)$ of the system is an infinite sequence of configurations, where $c_0$ is the initial configuration
%\footnote{Unless stated otherwise, we make no specific assumption regarding the respective positions of robots in initial configurations.} 
of the system, and every transition $c_i \rightarrow c_{i+1}$ is associated to the execution of a non empty subset of \emph{actions}. The granularity (or atomicity) of those actions is model-dependent and is defined in the sequel of this section. 

A \emph{scheduler} is a predicate on computations, that is, a scheduler defines a set of \emph{admissible} computations, such that every computation in this set satisfies the scheduler predicate. A \emph{scheduler} can be seen as an entity that is external to the system and selects robots for execution. As more power is given to the scheduler for robot scheduling, more different executions are possible and more difficult it becomes to design robot algorithms. In the remainder of the paper, we consider that the scheduler is fair and \emph{fully asynchronous}, that is, in any infinite execution, every robot is activated infinitely often, but there is no bound on the ratio between the most activated robot and the least activated one. In each cycle, the scheduler determines the distance to which each robot can move in this cycle, that is, it can stop a robot before it reaches its computed destination. However, a robot $r_i$ is guaranteed to be able to move a distance of at least $\delta_i$ towards its destination before it can be stopped by the scheduler.

We now review the main differences between the ATOM~\cite{SY99} and CORDA~\cite{FPSW99h} models. In the ATOM model, whenever a robot is activated by the scheduler, it performs a \emph{full} computation cycle. Thus, the execution of the system can be viewed as an infinite sequence of rounds. In a round one or more robots are activated by the scheduler and perform a computation cycle. The \emph{fully-synchronous ATOM} model refers to the fact that the scheduler activates all robots in each round, while the regular \emph{ATOM} model enables the scheduler to activate only a subset of the robots.
In the CORDA model, robots may be interrupted by the scheduler after performing only a portion of a computation cycle. In particular, phases (Look, Compute, Move) of different robots may be interleaved. For example, a robot $a$ may perform a Look phase, then a robot $b$ performs a Look-Compute-Move complete cycle, then $a$ computes and moves based on its previous observation (that does not correspond to the current configuration anymore). As a result, the set of executions that are possible in the CORDA model are a strict superset of those that are possible in the ATOM model. So, an impossibility result that holds in the ATOM model also holds in the CORDA model, while an algorithm that performs in the CORDA model is also correct in the ATOM model. Note that the converse is not necessarily true. 

\paragraph{\textbf{The RoboCast Problem}}

The RoboCast communication abstraction provides a set of robots located at arbitrary positions in a two-dimensional space the possibility to  broadcast their local information to each other. The RoboCast abstraction offers robots two communication primitives: 
$\emph{RoboCast(M)}$ sends \emph{Message} M to all other robots, 
and $\emph{Deliver(M)}$ delivers \emph{Message} M to the local robot. The message may consists in the local coordinate system, the robot chirality, the unit of measure, or any binary coded information.

%If a robot in the system invokes \emph{RoboCast(m)} at some time $t$ then the following two properties are satisfied:
Consider a run at which each robot $r_i$ in the system invokes \emph{RoboCast($m_i$)} at some time $t_i$ for some message $m_i$.
Let $t$ be equal to $max\{t_1, \ldots, t_n\}$.
% and let $m$ denotes the set $\{m_1, \ldots, m_n\}$. 
Any protocol solving the RoboCast Problem has to satisfy the following two properties: 

\begin{list}{}{}
\item \textit{Validity}: For each message $m_i$, there exists a time $t_i^\prime > t$ after which every robot in the system has performed \emph{Deliver($m_i$)}.
\item \textit{Termination}: There exists a time $t_T \geq max\{t_1^\prime, \ldots, t_n^\prime\}$ after which no robot performs a movement that causally depends on the invocations of \emph{RoboCast($m_i$)}.
\end{list}

\section{Local Coordinate System RoboCast}
\label{sec:lcs}

In this section we present algorithms for robocasting the local coordinate system.
For ease of presentation we first propose an algorithm for two-robots then the general version for systems with $n$ robots. 

The local coordinate system is defined by two axes (abscissa and ordinate), their positive directions and the unity of measure. 
In order to robocast this information we use a modular approach.
That is, robots invoke first the robocast primitive ($LineRbcast1$ hereafter) to broadcast a line representing their abscissa.
Then, using a parametrized module ($LineRbcast2$), they robocast three successive lines encoding respectively their ordinate, unit of measure and the positive direction of axes. 
%then they robocast the second 
%axis via a parametrized module. Finally using the same module the unity of measure is robocast. 
%This  invocation chain is motivated by the relation between the two axes of a coordinate system. When a node broadcasts a line, without any additional knowledge, two different points have to be sent in order to uniquely identify the line at the destination. However, in the case of a coordinate system, only for the first transmitted axis nodes need to identify the two points. The transmission of the second axis needs the knowledge of a unique additional point.
This invocation chain is motivated by the dependence between the transmitted lines.
When a node broadcasts a line, without any additional knowledge, two different points have to be sent in order to uniquely identify the line at the destination. However, in the case of a coordinate system, only for the first transmitted axis nodes need to identify the two points. The transmission of the subsequent axes needs the knowledge of a unique additional point.

\subsection{Line RoboCast}
In robot networks the broadcast of axes is not a new issue.
Starting with their seminal paper \cite{SY99}, Suzuki and Yamashita
presented an algorithm for broadcasting the axes via motion that works in the ATOM model. 
Their algorithm heavily relies on the atomicity of cycles and the
observation focus on the different positions of the other robots
during their Move phase.
%in different positions.
%The idea is as follows: 
%assume two robots $r$ and $r^\prime$ that want to broadcast their x-axes to each others.
%To do so, each of them move through the positive direction of its local x-axis each time it is activated. 
%Hence, both move on the positive direction of their x-axes.
%When a robot $r$ observes $r^\prime$ in at least four different positions, 
%it can infer that $r^\prime$ executed at least two complete cycles \cite{SY99b} which means that 
%$r^\prime$ observed $r$ in at least two different positions of its x-axis.
%Such a claim is true only in the atomic model ATOM which guarantees that robots execute their cycles in a lock-step manner. 
%The same idea have been exploited in \cite{DDPS09} for building stigmergic systems in ATOM model. 

%then $r$ is sure that it was observed by $r^\prime$ in at least two different positions which allows it to infer that its line has been received by its received.
%Robot $r$ can make such a conclusion because  
%This is because $r$ knows that  $r^\prime$ executed at least two complete cycles. This 
%is possible only in the ATOM model which guarantees that robots execute their cycles in a lock-step manner. 
%At this point, $r$ has received the x-axis of its peer and know that its own x-axis has also been received by its peer.

This type of observation is totally useless in asynchronous CORDA model. 
In this model, when a robot $r$ moves towards its destination, another robot $r^\prime$ can be activated $k>1$ times with $k$ arbitrarily large, and thus observe $r$ in $k$ different positions without having any clue on the number of complete cycles executed by $r$. In other words, the number of different positions observed for a given robot is not an indicator on the number of complete executed cycles since in CORDA cycles are completely uncorrelated.

Our solution uses a novel strategy. 
That is, the focus moves from observing robots in different positions
to observing their change of direction: each robot changes 
its direction of movement when a particular stage of the algorithm is
completed; this change allows the other robots to infer 
information about the observed robot.

\subsubsection{Line RoboCast Detailed Description}
Let $r_0$ and $r_1$ be the two robots in the system.
In the sequel, when we refer to one of these robots without specifying which, we denote it by $r_i$ and its peer by $ r_{1-i}$.
In this case, the operations on the indices of robots are performed modulo 2.
%The goal of our algorithm is that each robot $r_i$ simultaneously broadcasts to the other a privately chosen line $l_i$.
For ease of presentation we assume that initially each robot $r_i$ translates and rotates its local coordinate system 
such that its x-axis and origin coincide with the line to be broacast and its current location respectively.
We assume also that each robot is initially located in the origin of its local coordinate system.

At the end of the execution, 
each robot must have broadcast its own line and have received the line of its peer.
%The algorithm must also provide an acknowledgement mechanism so that the two robots be sure, 
%before finishing their execution, that their respective lines have been received.
A robot "receives" the line broadcast by its peer when it knows at least two distinct positions of this line. 
%That is, when it observes its peer in at least two distinct positions.
Thus, to send its line, each robot must move along it (following a scheme that will be specified later) until it is sure that it has been observed by the other robot.
 
The algorithm idea is simple:
each robot broadcasts its line by moving along it in a certain direction (considered to be positive). Simultaneously, it observes the different positions occupied by its peer $r_{1-i}$.
Once $r_i$ has observed $r_{1-i}$ in two distinct positions, 
it informs it that it has received its line by changing its direction of movement, 
that is, by moving along its line in the reverse direction (the negative direction if the first movement have been performed in the positive direction of the line).
This change of direction is an acknowledgement for the reception of the peer line.
%This is also a kind of discretization of time, because a robot that observes a change of direction of its peer, 
%concludes that the latter has advanced in the execution of its algorithm and has executed at least one complete cycle.
A robot finishes the algorithm once it changed its direction and observed that the other robot also changed its direction. 
This means that both robots have sent their line and received the other's line.

The algorithm is described in detail as Algorithm \ref{alg-getview2}. Due to space limitations, its proof is given in the Appendix.
Each robot performs four stages referred in Algorithm \ref{alg-getview2} as states:

\begin{itemize}
\item state $S_1$:  This is the initial state of the algorithm. 
At this state, the robot $r_i$ stores the position of its peer in the variable $pos_1$
and heads towards the position $(1.0)$ of its local coordinate system. That is,
it moves along its line in the positive direction. 
Note that $r_i$ stays only one cycle in this state and then goes to state $S_2$.

\item state $S_2$: 
A this point, $r_i$ knows only one point of its peer line (recorded in $pos_1$).
To be able to compute the whole peer line, $r_i$ must observe $r_{1-i}$ in another (distinct) position of this line.
Hence, each time it is activated, $r_i$ checks if $r_{1-i}$ is still located in $pos_1$ or if it has already changed its position. 
In the first case (line $2.a$ of the code), it makes no movement by selecting its current position as its destination.
Otherwise (line $2.b$), 
it saves the new position of $r_{1-i}$ in $pos_2$ and delivers the line formed by $pos_1$ and $pos_2$.
Then, it initiates a change of direction by moving towards the point $(-1.0)$ of its local coordinate system, 
and moves to state $S_3$.

\item state $S_3$:
at this point $r_i$ knows the line of its peer locally derived from $pos_1$ and $pos_2$.
Before finishing the algorithm, $r_i$ must be sure that also $r_{1-i}$ knows its line.
Therefore, it observes $r_{1-i}$ until it detects a change of direction (the condition of line $3.a$).
If this is not the case and if $r_i$ is still in the positive part of its x-axis, then it goes to the position $(-1, 0)$ of its local coordinate system (line $3.b$).
Otherwise (if $r_i$ is already in the negative part of its x-axis), it performs a null movement (line $3.c$).
When $r_i$ is in state $S_3$ one is sure, as we shall show later,  that $r_{1-i}$ knows at least one position of $l_i$, say $p$.
Recall that $l_i$ corresponds to the x-axis of $r_i$. It turns out that $p$ is located in the positive part of this axis.
In moving towards the negative part of its x-axis, 
$r_i$ is sure that it will eventually be observed by $r_{1-i}$ in a position distinct from $p$ which allows $r_{1-i}$ to compute $l_i$.

\item state $S_4$:
At this stage, both $r_i$ and $r_{1-i}$ received the line sent by each others.
That is, $r_i$ has already changed its own direction of movement, and observed that $r_{1-i}$ also changed its direction.
But nothing guarantees that at this step $r_{1-i}$ knows that $r_i$ changed its direction of movement.
If $r_i$ stops now, $r_{1-i}$ may remain stuck forever (in state $S_3$).
To announce the end of the algorithm to its peer, $r_i$ heads towards a position located outside $l_i$,
% (line $4.a$).
That is, it will move on a line $nextl_i$ (distinct from $l_i$) which is
given as parameter to the algorithm.
During the move from $l_i$ to $nextl_i$, $r_i$
should avoid points outside these lines. 
To this end, $r_i$ must first pass through $myIntersect$ - which is the intersection of $l_i$ and $nextl_i$ -
before moving to a point located in $nextl_i$ but not on $l_i$ (refer to lines $3.a.2$, $3.a.3$ and $4.a$ of the code).
%$nextl_i$ can be any line that is different from $l_i$.

Note that the robocast of a line is usually followed by the robocast
of other information (e.g. other lines that encode the local
coordinate system). To helps this process the end of the robocast of $l_i$ 
should mark the beginning of the next line, $nextl_i$, robocast.
Therefore, once $r_i$ reaches $myIntersect$,
$r_i$ rotates its local coordinate system such that its x-axis matches now with $nextl_i$, 
and then it moves toward the point of (1,0) of its (new) local coordinate system.
When $r_{1-i}$ observes $r_i$ in a position that is not on $l_i$, it learns that $r_i$ knows that $r_{1-i}$ learned $l_{1-i}$, 
and so it can go to state $S_4$ (lines $3.a.*$) and finish the algorithm.
\end{itemize}

\begin{algorithm}[htb]
\begin{tiny}

\caption{Line RoboCast \bf{LineRbcast1} for two robots: Algorithm for robot $r_i$}          
\label{alg-getview2}                  
\begin{algorithmic}

%\STATE   \textbf{Init}:\\
%Initially each robot is located on line $l_i$ and uses a local coordinate system based on %it. 
%The origin is the current location and the X-axis is the line $l_i$.
%The positive direction of the X-axis is meaningless.\\
%\STATE

\STATE \textbf{Variables}:\\

\STATE $state$: initially $S_1$
\STATE $pos_1, pos_2$: initially $\perp$
\STATE $destination, myIntersect$: initially $\perp$
\STATE

\STATE   \textbf{Actions}:\\

\STATE \textbf{1. State [$S_1$]}: \emph{\%Robot $r_i$ starts the algorithm\%}
$\vspace{0.2cm}$
\STATE $\hspace{1cm}$a. $pos_1 \leftarrow observe(1-i)$
%\STATE a. \begin{tabbing} \vspace{2cm}\=\kill\> \end{tabbing} $pos_1 \leftarrow observe()$.
\STATE $\hspace{1cm}$b. $destination \leftarrow (1,0)_i$
\STATE $\hspace{1cm}$c. $state \leftarrow S_2$
\STATE $\hspace{1cm}$d. Move to destination
\STATE

\STATE \textbf{2. State [$S_2$]}: \emph{\%$r_i$ knows one position of $l_{1-i}$\%}
$\vspace{0.2cm}$
\STATE $\hspace{1cm}$a. \textbf{if} ($pos_1=observe(1-i)$) \textbf{then} $destination \leftarrow observe(i)$
%\STATE $\hspace{2cm}$ 1. $destination \leftarrow observe(i)$
\STATE $\hspace{1cm}$b. \textbf{else} 
\STATE $\hspace{2cm}$ 1. $pos_2 \leftarrow observe(1-i)$
\STATE $\hspace{2cm}$ 2. $l_{1-i} \leftarrow line(pos_1, pos_2)$
\STATE $\hspace{2cm}$ 3. Deliver ($l_{1-i}$) 
\STATE $\hspace{2cm}$ 4. $destination \leftarrow (-1,0)_i$
\STATE $\hspace{2cm}$ 5. $state \leftarrow S_3$ \textbf{endif}
\STATE $\hspace{1cm}$c. Move to destination
\STATE

\STATE \textbf{3. State [$S_3$]}: \emph{\%$r_i$ knows the line robocast by robot $r_{1-i}$\%}
$\vspace{0.2cm}$
\STATE $\hspace{1cm}$a. \textbf{if} ($pos_2$ is not inside the line segment $[pos_1, observe(1-i)]$) \textbf{then}
\STATE $\hspace{2cm}$ 1. $state \leftarrow S_4$
\STATE $\hspace{2cm}$ 2. $myIntersect \leftarrow intersection(l_i, nextl_i)$
\STATE $\hspace{2cm}$ 3. $destination \leftarrow myIntersect$
\STATE $\hspace{1cm}$b. \textbf{else if} ($observe(i) \geq (0,0)_i$) \textbf{then} $destination \leftarrow (0,-1)_i$
\STATE $\hspace{1cm}$c. \textbf{else} $destination \leftarrow observe(i)$ \textbf{endif endif}
\STATE $\hspace{1cm}$d. Move to destination
\STATE

\STATE \textbf{4. State [$S_4$]}: \emph{\%$r_i$ knows that robot $r_{1-i}$ knows its line $l_i$\%}
$\vspace{0.2cm}$
%\STATE $\hspace{1cm}$a. $intersectp \leftarrow intersection(l_i, nextl_i)$
\STATE $\hspace{1cm}$a. \textbf{if} ($observe(i) \neq myIntersect$) \textbf{then} $destination \leftarrow myIntersect$
\STATE $\hspace{1cm}$b. \textbf{else}
\STATE $\hspace{2cm}$1. $r_i$ rotates its coordinate system such that its x-axis and the origin match with
\STATE $\hspace{2cm}$	 $nextl_i$ and $myIntersect$ respectively.
\STATE $\hspace{2cm}$2. $destination \leftarrow (1,0)_i$; return \textbf{endif}
%\STATE $\hspace{2cm}$3.  return $l_{1-i} \leftarrow line(pos_1, pos_2)$ \textbf{endif}
\STATE $\hspace{1cm}$c. Move to destination
\STATE

\remove{
\STATE \textbf{4. State [$S_4$]}: \emph{\%$r_i$ knows that $r_{1-i}$ knows its $l_i$\%}
\STATE $\hspace{1cm}$a. $destination \leftarrow (observe(i).x, y)_i$ with $y \neq 0$
\STATE $\hspace{1cm}$b. Move to destination
}

\end{algorithmic}

\end{tiny}
\end{algorithm}

\subsection{Line RoboCast: a Composable Version}
Line RoboCast primitive is usually used as a building block for
achieving more complex tasks.
For example, the RoboCast of the local coordinate system requires the
transmission of four successive lines representing respectively
the abscissa, the ordinate, the value of the unit 
measure and a forth line to determine the positive direction of axes.
In stigmergic communication a robot has to transmit at least a line
for each binary information it wants to send.
In all these examples, the transmitted lines are dependent one of each
other and therefore their successive transmission can be accelerated by
directly exploiting
this dependence.
%In this section we present some modifications to the previous algorithm to make it composable.
%The first modification we introduce takes advantage of this observation to speed the robocast operation.
Indeed, the knowledge of a unique point (instead of two) is sufficient for the receiver to infer the sent line.
%It can do this by exploiting the relationship that binds this line with the lines it has already received earlier.
In the following we propose modifications of the Line RoboCast primitive 
in order to exploit contextual information that are encoded in a set
of predicates that will be detailed in the sequel. 

In the case of the local coordinate system, the additional information
the transmission can exploit is the fact that the abscissa is perpendicular to the ordinate.
Once the abscissa is transmitted, it suffices for a robot to simply send a single position of its ordinate, say $pos1$. 
The other robots can then calculate the ordinate by finding the line that passes through $pos1$ and which is perpendicular to the previously received abscissa.
In the modified version of the Line RoboCast algorithm the predicate $isPerpendicular$ encodes this condition.

For the case of stigmergy, a robot transmits a binary information by robocasting a line whose angle to the abscissa encodes this information.
The lines transmitted successively by a single robot are not perpendicular to each others.
However, all these lines pass through the origin of the coordinate system of the sending robot.
In this case, it suffices to transmit only one position located on this line as long as it is distinct from the origin. 
We say in this case that the line satisfies the predicate $passThrOrigin$.

A second change we propose relates to the asynchrony of the algorithm.
In fact, even if robots execute in unison, 
they are not guaranteed to finish the execution of $LineRbcast1$ at the same time (by reaching $S_4$).
A robot $ r_i $ can begin transmitting its $k$-th line $ l_i $ when its peer $r_{1-i}$
is still located in its $(k-1)$-th line $ancientl_{1-i}$ that $r_i$ has already received. 
$r_i$ should ignore the positions transmitted by $r_{1-i}$ until it leaves $ancientl_{1-i}$ for a new line.
It follows that to make the module composable, the old line
that the peer has already received from 
its peer should be supplied as an argument ($ancientl_{1-i}$) to the function.
Thus, it will not consider the positions occupied by $r_ {1-i}$ until the latter leaves $ancientl_{1-i}$.

In the following, we present the code of the new Line RoboCast function that we denote by $LineRbcast2$.
Its description and its formal proof are omitted since they follow the same lines as those of $LineRbcast1$.

\begin{algorithm}[htb]
\begin{tiny}
\caption{Line RoboCast \bf{LineRbcast2} for two robots: Algorithm for robot $r_i$}          
\label{alg-viewmulticast2}                  
\begin{algorithmic}

\remove{
\STATE \textbf{Init}:\\
Initially each robot is located on line $l_i$ and uses a local coordinate system based on it. 
The origin is the current location and the X-axis is the line $l_i$.
The positive direction of the X-axis is meaningless.\\
\STATE
}

\STATE \textbf{Inputs:}
\STATE $l_i$ : the line to robocast
\STATE $nextl_i$: the next line to robocast after $l_i$
\STATE $precedentl_{1-i}$: the line robocast precedently by $r_{1-i}$
\STATE $predicate$: a predicate on the output $l_{1-i}$, for example $isPerpendicular$ and $passThrOrigin$.
\STATE

\STATE \textbf{Outputs:}
\STATE $l_{1-i}$ : the line robocast by $r_{1-i}$
\STATE

\STATE \textbf{Variables}:
\STATE $state$: initially $S_1$
\STATE $pos_1$: initially $\perp$
\STATE $destination, myIntersect, peerIntersect$: initially $\perp$
\STATE

\STATE   \textbf{Actions}:\\

\STATE \textbf{1. State [$S_2$]}: \emph{\%$r_i$ starts robocasting its line $l_i$\%}
% \emph{\%Robot $r_i$ knows one position of the X-axis of robot $r_{1-i}$\%}
$\vspace{0.2cm}$
\STATE $\hspace{1cm}$a. \textbf{if} ($observe(1-i) \in precedentl_{1-i}$) \textbf{then} $destination \leftarrow observe(i)$
%\STATE $\hspace{2cm}$ 1. $destination \leftarrow observe(i)$
\STATE $\hspace{1cm}$b. \textbf{else} 
\STATE $\hspace{2cm}$ 1. $pos3 \leftarrow observe(1-i)$
\STATE $\hspace{2cm}$ 2. $l_{1-i} \leftarrow$ the line that passes through $pos3$ and satisfies $predicate$.
\STATE $\hspace{2cm}$ 3. Deliver ($l_{1-i}$)
\STATE $\hspace{2cm}$ 4. $peerIntersect \leftarrow$ intersection between $l_{1-i}$ and $precedentl_{1-i}$
%compute the value of variable $origin$ as the intersection of $]pos_1, pos_2[$ and the line which is perpendicular to it and which passes through $pos3$.
\STATE $\hspace{2cm}$ 5. $destination \leftarrow (0, -1)_i$
\STATE $\hspace{2cm}$ 6. $state \leftarrow S_3$ \textbf{endif}
\STATE $\hspace{1cm}$c. Move to destination
\STATE

\STATE \textbf{2. State [$S_3$]}: \emph{\%$r_i$ knows the line robocast by robot $r_{1-i}$\%}
$\vspace{0.2cm}$
%\STATE $\hspace{1cm}$a. \textbf{if} ($origin$ is inside the segment $[pos3, observe(1-i)]$) or ( \textbf{then}
\STATE $\hspace{1cm}$a. \textbf{if} ($pos3$ is not inside the line segment $[peerIntersect, observe(1-i)]$) \textbf{then}
\STATE $\hspace{2cm}$ 1. $state \leftarrow S_4$
\STATE $\hspace{2cm}$ 2. $myIntersect \leftarrow intersection(l_i, nextl_i)$
\STATE $\hspace{2cm}$ 3. $destination \leftarrow myIntersect$
\STATE $\hspace{1cm}$b. \textbf{else if} ($observe(i) \geq (0,0)_i$) \textbf{then} $destination \leftarrow (0,-1)_i$
%\STATE $\hspace{2cm}$ 1. $destination \leftarrow (-1,0)$
\STATE $\hspace{1cm}$c. \textbf{else} $destination \leftarrow observe(i)$ \textbf{endif endif}
%\STATE $\hspace{2cm}$ 1. $destination \leftarrow observe(i)$
\STATE $\hspace{1cm}$d. Move to destination
\STATE

\STATE \textbf{3. State [$S_4$]}: similar to state $S_4$ of the $lineRbcast1$ function.
\STATE

\remove{
\STATE \textbf{3. State [$S_3$]}: \emph{\%$r_i$ knows that $r_{1-i}$ knows its line $l_i$\%}
%\STATE $\hspace{1cm}$a. $intersectp \leftarrow intersection(l_i, nextl_i)$
\STATE $\hspace{1cm}$a. \textbf{if} ($observe(i) \neq myIntersect$) \textbf{then} $destination \leftarrow myIntersect$
\STATE $\hspace{1cm}$b. \textbf{else}
\STATE $\hspace{2cm}$1. $r_i$ rotates its coordinate system such that its x-axis and the origin match with
\STATE $\hspace{2cm}$	 $nextl_i$ and $myIntersect$ respectively.
\STATE $\hspace{2cm}$2. $destination \leftarrow (1,0)_i$
\STATE $\hspace{2cm}$3.  return $l_{1-i}$ \textbf{endif}
\STATE $\hspace{1cm}$c. Move to destination
\STATE
}

\remove{
\STATE \textbf{7. State [$S_7$]}: \emph{\%Robot $r_i$ knows that robot $r_{1-i}$ knows its Y-axis\%}
$\vspace{0.2cm}$
\STATE $\hspace{1cm}$a. \textbf{if} ($(observe(i).x=0)$ and $(observe(i) < (0,0)_i)$) $destination \leftarrow (0,0)_i$ 
\STATE $\hspace{1cm}$b. \textbf{else} 
\STATE $\hspace{2cm}$ 1.$destination \leftarrow$ any point not located in the Y-axis of $r_i$ (or start the diffusion of another line)
\STATE $\hspace{2cm}$ 2.$state \leftarrow S_8$ \textbf{endif}
%\STATE $\hspace{1cm}$b. \textbf{else} $destination \leftarrow (0, y)_i$ with $y \neq 0$
\STATE $\hspace{1cm}$c. Move to destination
\STATE
}
\end{algorithmic}

\end{tiny}
\end{algorithm}

%\newpage

\subsection{RoboCast of the Local Coordinate System}
\label{subsec:coord-rbcast}
To robocast their two axes (abscissa and ordinate), robots call LineRbcast1 to robocast the abscissa, then LineRbcast2 to robocast the ordinate.
The parameter $\neq myOrdinate$ of $LineRbcast2$ stands for the next line to be robocast  and it can be set to any line different from $myOrdinate$.
The next line to robocast ($unitLine$) is a line whose angle with the
x-axis encodes the unit of measure. This angle will be determined
during 
%Infortunately, this angle cannot be computed before calling $LineRbcast2$ in line 2, that is before receiving the peer ordinate. 
%Hence, we bend the modularity rule by not giving $unitLine$ as a parameter to the function $LineRbcast2$ of line 2 
%and we determine it during 
the execution $LineRbcast2$.

\begin{quote}
1. $peerAbscissa \leftarrow LineRbcast1 (myAbscissa, myOrdinate)$\\
2. $peerOrdinate \leftarrow$\\
$LineRbcast2 (myOrdinate, \neq myOrdinate, peerAbscissa, isPerpendicular)$
\end{quote}

After executing the above code, each robot knows the two axes of its
peer coordinate 
system but not their positive directions neither their unit of measure.
To robocast the unit of measure we use a technique similar to that used by \cite{SY99}. 
The idea is simple: each robot measures the distance $d_i$ between its origin and the peer's origin in terms of its local coordinate system.
To announce the value of $d_i$ to its peer, each robot robocast via LineRbcast2 a line, $unitLine$, 
which passes through its origin and whose angle with its abscissa is equal to $f(d_i)$
where for $x>0, f(x)=(1/2x) \times 90^\circ$ is a monotonically increasing function with range $(0^\circ,90^\circ)$. 
The receiving robot $r_{1-i}$ can then infer $d_i$ from $f(d_i)$ 
and compute the unit measure of $r_i$ which is equal to $d_{1-i}/d_i$.
The choice of $(0^\circ,90^\circ)$ as a range for $f(x)$ (instead of  $(0^\circ,360^\circ)$) 
is motivated by the fact that the positive directions of the two axes are not yet known to the robots.
It is thus impossible to distinguish between an angle $\alpha$ with $\alpha \in (0^\circ,90^\circ)$ and the angles 
$\Pi-\alpha$, $-\alpha$, and $\Pi+\alpha$.
To overcome the ambiguity and to make $f(x)$ injective, we restrict
the range to $(0^\circ,90^\circ)$. 
In contrast, Suzuki and Yamashita \cite{SY99} use a function $f^\prime(x)$ 
slightly different from ours: $(1/2x) \times 360^\circ$.
That is, its range is equal to $(0^\circ,360^\circ)$.
This is because in ATOM, robots can robocast at the same time the two axis and their positive directions, 
for example by restricting the movement of robots to only the positive part of their axes.
Since the positive directions of the two axes are known, $unitLine$ can be an oriented line whose angle $f^\prime(x)$ 
can take any value in $(0^\circ,360^\circ)$ without any possible ambiguity.
%Such a restriction is impossible to impose in our robocast routines because 
%robots need to have the capability to change their direction to move in both the positive and the negative part of their axes, 
%Note that in \cite{SY99}, the authors use a function $f^\prime(x)$ slightly different from ours ($f(x)$) and which is equal to $(1/2x) \times 90^\circ$.

%This is because in ATOM, robots can robocast their axes and their positive directions at the same time, for example by moving only through the positive part of their axes. Since we use change of directions instead of changes of positions to robocast lines, we can only robocast the axis without their positive direction because each robot move through both the positive and the negative parts of its axes.
%Hence, the range of $f^\prime(x)$ in \cite{SY99} is $(0^\circ,360^\circ)$ because it corresponds to the angle of a directed line and since the positive directions of the the axes are known, $f^\prime(x)$ is bijective.
%In our case, we cannot distinguish for example between an angle $\alpha$ with $\alpha \in (0^\circ,90^\circ)$ and the angles 
%$\Pi-\alpha$, $-\alpha$, and $\Pi+\alpha$, this is why we restrict the range of $f(x)$ to $(0^\circ,90^\circ)$ to make it injective. 
%(while $f^\prime(x)$ is not in our case).

\paragraph{\textbf{Positive directions of axes}}
Once the two axes are known, determining their positive directions amounts 
to selecting the upper right quarter of the coordinate system that is positive for both $x$ and $y$.
Since the line used to robocast the unit of distance passes through two quarters (the upper right and the lower left), 
it remains to choose among these two travelled quarters which one corresponds to the upper right one.
To do this, each robot robocast just after the line encoding the unit
distance another line which 
is perpendicular to it such that their intersection lays inside the upper right quarter.

\paragraph{\textbf{Generalization to $n$ robots}}
The generalization of the solution to the case of $n> 2$  robots has to use an additional mechanism to allow robots to "recognize" other robots 
and distinguish them from each others despite anonymity.
Let us consider the case of three robots
$r_1,r_2,r_3$. When $r_1$ looks the second time, $r_2$ and $r_3$ could have moved
(or be moving), each according to its local coordinate system and unit measure.
At this point, even with memory of  past observations, $r_1$ may be not able to distinguish between $r_2$ and $r_3$
 in their new positions given the fact that robots are anonymous.
Moreover, $r_2$ and $r_3$ could even switch places and appear
not to have moved.
Hence, the implementation of the primitive $observe(i)$ is not trivial. 
For this, we use the collision avoidance techniques presented in the next section to instruct each robot to move only in the vicinity of its initial position.
This way, other robots are able to recognize it by using its past positions. The technical details of this mechanism are given at the end of the next section.

Apart from this, the generalization of the protocol with $n$ robots is trivial. We present its detailed description in the Appendix.

\subsection{Motion Complexity Analysis}

Now we show that the total number of robot moves in the
coordinate system RoboCast is upper bounded. 
For the sake of presentation, we assume for now that the scheduler
does not interrupt robots execution before they reach their planned
destination. Each robot is initially located at the origin of its
local coordinate system. To robocast each axis, a robot must visit two
distinct positions: one located in the positive part of this axis and the other one
located in its negative part. For example, 
to robocast its $x$-axis, a robot has first to move from its origin to
the position $(1.0)_i$, then from $(1.0)_i$ 
to the $(-1, 0)_i$. Then, before initiating a robocast for the other
axis, the robot must first return back to its origin. 
Hence, at most 3 movements are needed to robocast each axis. This
implies that to robocast the whole local coordinate system, 
at most 12 movements have to be performed by a particular robot.

In the general CORDA model, the scheduler is allowed to stop robots before
they reach their destination, as long as a minimal distance 
of $\delta_i $ has been traversed. In this case, the number of
necessary movements is equal to at most $8 * (1 +1 / \delta_i)$. 
This worst case is obtained when a robot is \emph{not} stopped by the
scheduler when moving from its origin towards another position 
(thus letting it go the farthest possible), but stopped 
whenever possible when returning back from this (far) position to the origin.

This contrasts with \cite{SY99} and \cite{DDPS09} where the number of
positions visited by each robot to robocast a line is unbounded 
(but finite). This is due to the fact that in both approaches, robots
are required to make a non null movement whenever activated until 
they know that their line has been received. Managing an arbitrary
large number of movements in a restricted space to prevent collisions 
yields severe requirements in~\cite{DDPS09}: either robots are allowed
to perform infinitely small movements (and such movements can be 
seen by other robots with infinite precision), or the scheduler is
restricted in its choices for activating robots (no robot can be 
activated more than $k$ times, for a given $k$, between any two
activations of another robot) and yields to a setting that is not 
fully asynchronous. Our solution does not require any such hypothesis.

\section{Collision-free RoboCast}
\label{sec:collisionfree}
In this section we enhance 
the algorithms proposed in Section \ref{sec:lcs} with the collision-free feature.
In this section we propose novel techniques for collision avoidance 
that cope with the system asynchrony.

Our solution is based on the same principle of locality as the Voronoi
Diagram based schemes. 
However, acceptable moves for a robot use a different geometric
area. This area is defined for 
each robot $r_i$ as a local \emph{zone of movement} and is denoted by
$ZoM_i$.
We require that 
each robot $r_i$ moves only inside $ZoM_i$. The intersection of
different $ZoM_i$ must remain empty at 
all times to ensure collision avoidance. We now present three possible
definitions for the zone of movement: $ZoM^1_i$, $ZoM^2_i$ and
$ZoM^3_i$. All three ensure collision avoidance in CORDA, but only the
third one can be computed in a model where 
robots do not know the initial position of their peers.

Let $P(t)=\{p_1(t), p_2(t) \ldots, p_n(t)\}$ be the configuration of
the network at time $t$, such that 
$p_i(t)$ denotes the position of robot $r_i$ at time $t$ expressed
in a global coordinate system. 
This global coordinate system is unknown to individual robots and is
only used to ease the presentation and the proofs. 
Note that $P(t_0)$ describes the initial configuration of the network.

\begin{definition}(\textbf{Voronoi Diagram})\cite{A91}
\label{defca:voronoi}
The Voronoi diagram of a set of points $P = \{p_1, p_2, \dots , p_n\}$ is a subdivision of the plane into $n$ cells, one for each point in $P$. The cells have the property that a point $q$ belongs to the Voronoi cell of point $p_i$ iff for any other point $p_j \in P$, $dist(q,p_i) < dist(q,p_j)$ where $dist(p,q)$ is the Euclidean distance between $p$ and $q$. In particular, the strict inequality means that points located on the boundary of the Voronoi diagram do not belong to any Voronoi cell.
\end{definition}

\begin{definition}(\textbf{$ZoM^1_i$})
\label{defca:zom1}
Let $DV(t_0)$ be the Voronoi diagram of the initial configuration $P(t_0)$. For each robot $r_i$, the zone of movement of $r_i$ at time $t$, $ZoM^1_i(t)$, is the Voronoi cell of point $p_i(t_0)$ in $DV(t_0)$.
\end{definition}

\begin{definition}(\textbf{$ZoM^2_i$})
\label{defca:zom2}
For each robot $r_i$, define the distance $d_i$ = min\{$dist(p_i(t_0), p_j(t_0))$ \textsl{ with } $r_j \neq r_i$\}. 
The zone of movement of $r_i$ at time $t$, $ZoM^2_i(t)$, is the circle centered in $p_i(0)$ and whose diameter is equal to $d_i/2$.
 A point $q$ belongs to $ZoM^2_i(t)$ iff $dist(q, p_i(t_0)) < d_i/2$.
\end{definition}

\begin{definition}(\textbf{$ZoM^3_i$})
\label{defca:zom3}
For each robot $r_i$, define the distance $d_i(t)$ = min\{$dist(p_i(t_0), p_j(t))$ \textsl{ with } $r_j\neq r_i$\} at time $t$. 
The zone of $r_i$ at time $t$, $ZoM^3_i(t)$, is the circle centered in $p_i(t_0)$ and whose diameter is equal to $d_i(t)/3$. 
A point $q$ belongs to $ZoM^3_i(t)$ iff $dist(q, p_i(t_0)) < d_i(t)/3$.
\end{definition}

\begin{figure}[htbp]
\begin{center}
\centering
\subfigure[{$ZoM^2_p$}]
{\includegraphics[scale=.55]{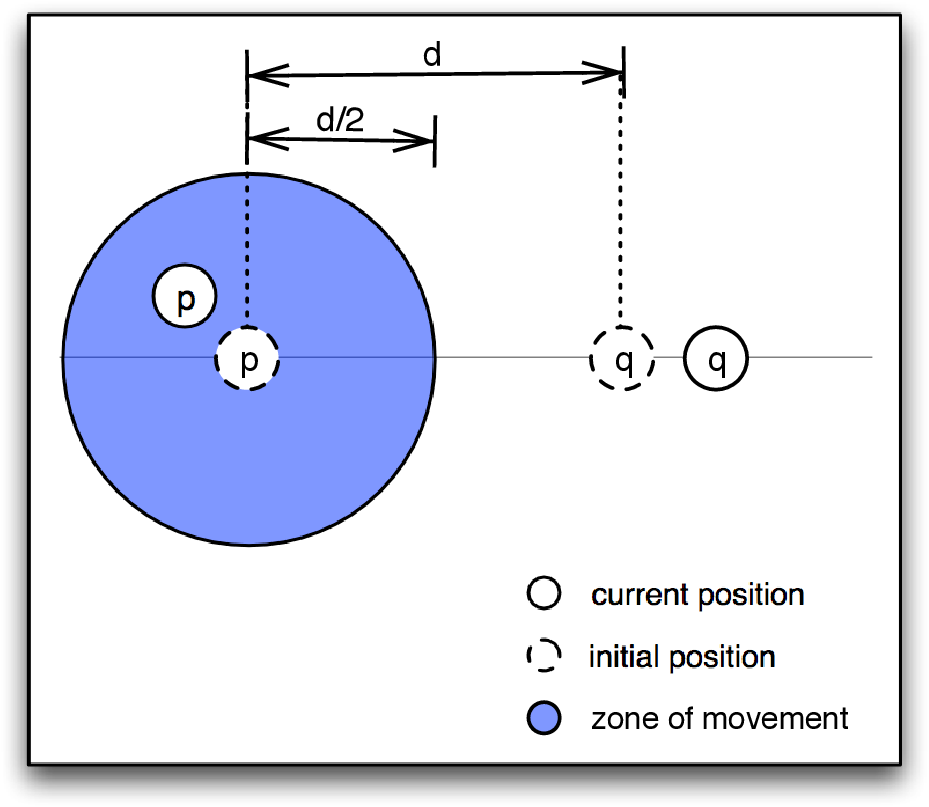}}
%\hspace{2cm}
\subfigure[{$ZoM^3_p$}]
{\includegraphics[scale=.55]{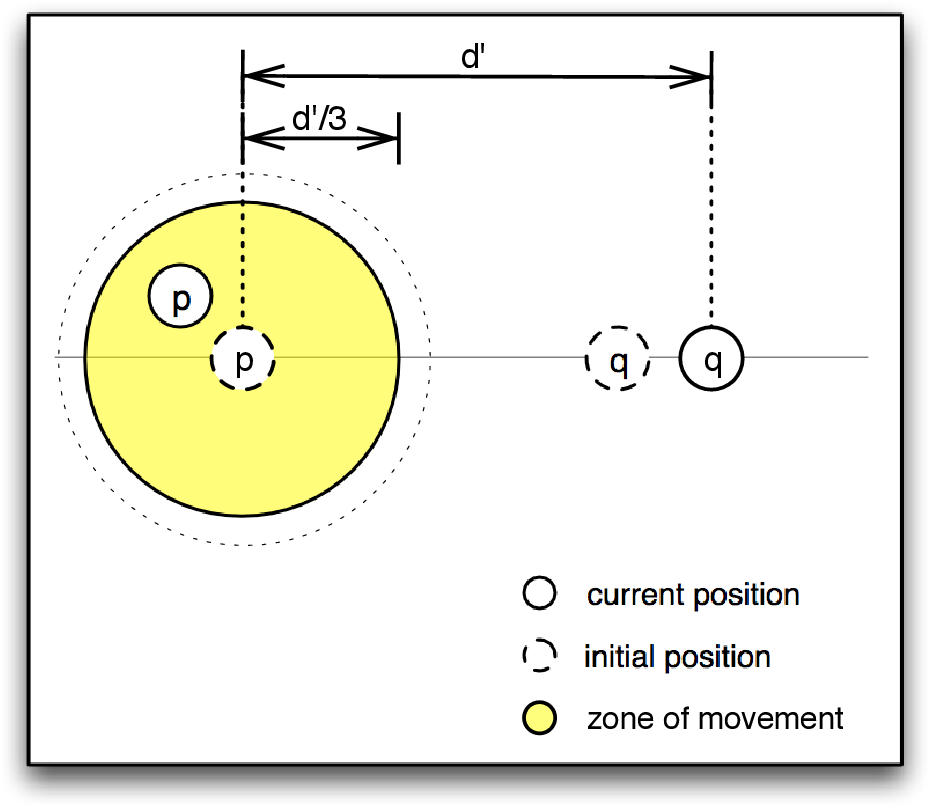}}

\caption{ \footnotesize Example zones of movement: The network is formed of two robots: $p$ and $q$. $d$ is the distance between the initial positions of $p$ and $q$ (dashed circles), $d^\prime$ is the distance between the initial position of $p$ and the current position of $q$. The diameter of $ZoM^2_p$ (blue) is $d/2$ and that of $ZoM^3_p$ (yellow) is $d^\prime/3$.}
\label{fig-zom3-collision}
\end{center}
\end{figure}

Note that $ZoM^1$ and $ZoM^2$ are defined using information about the initial configuration $P(t_0)$, and thus cannot be used with the hypotheses of Algorithm~\ref{alg-viewmulticast2}. In contrast, robot $r_i$ only needs to know its \emph{own} initial position and the \emph{current} positions of other robots to compute $ZoM^3_i$. As there is no need for $r_i$ to know the \emph{initial} positions of other robots, $ZoM^3_i$ can be used with Algorithm \ref{alg-viewmulticast2}. It remains to prove that $ZoM^3_i$ guarantees collision avoidance. We first prove that $ZoM^1_i$ does, which is almost trivial because its definition does not depend on time. Then, it suffices to prove that $ZoM^3_i \subseteq ZoM^2_i \subseteq ZoM^1_i$. Besides helping us in the proof, $ZoM^2_i$ can be interesting in its own as a cheap collision avoidance scheme in the ATOM model, as computing a cycle of radius half the distance to the nearest neighbor is much easier that computing a full blown Voronoi diagram.

\begin{lem}
\label{lemca:zom1}
If $\forall t$, for each robot $r_i$, the destination point computed by $r_i$ at $t$ remains inside $ZoM^1_i(t)$, then collisions are avoided.
\end{lem}

\begin{proof}
By definition of Voronoi diagram, different Voronoi cells do not overlap. Moreover, for a given $i$, $ZoM^1_i$ is static and does not change over time. Hence, $\forall i,j \in \Pi$, $\forall t, t^\prime$, $ZoM^1_i(t) \cap ZoM^1_j(t^\prime) = \emptyset$.
\end{proof}

Clearly, $ZoM^2_i \subseteq ZoM^1_i$ which means that $ZoM^2_i$ ensures also collision avoidance.

\begin{lem}
\label{lemca:zom2}
If $\forall t$, for each robot $r_i$, the destination point computed by $r_i$ at $t$ always remains inside $ZoM^2_i(t)$, then collisions are avoided.
\end{lem}

%\begin{proof}
The proof of the above lemma follows directly from the fact that $\forall t ZoM^2_i(t) \subseteq ZoM^1_i(t)$ and Lemma~\ref{lemca:zom1}.
%\end{proof}

\begin{lem}
\label{lemca:zom3}
$\forall t, ZoM^3_i(t) \subseteq ZoM^2_i(t)$. 
\end{lem}

\begin{proof}
Fix some robot $r_i$ and let $r_j$ be the closest robot from $r_i$ at time $t_0$. Let $d_0$ denote the initial distance between $r_i$ and $r_j$, that is, $d_0=dist(p_i(t_0), p_j(t_0))$. We assume that all robots move only inside their $ZoM^3_i$ computed as explained in Definition~\ref{defca:zom3}. Let $t_1 \geq t_0$ be the first time at which a robot in $\{r_i, r_j\}$, say \emph{e.g.} $r_i$, finishes a Look phase after $t_0$. The destination computed by $r_i$ in this cycle is located inside $ZoM^3_i(t_1)$, which is a circle centered at $p_i(t_0)$ and whose diameter is $\leq d_0/3$. Hence, the destination computed by $r_i$ is distant from $p_i(0)$ by at most $d_0/3$. Let $t_2 \geq t_1$ be the first time after $t_1$ at which a robot, say $r_j$, finishes a Look phase. Between $t_1$ and $t_2$, $r_i$ may have finished its Move phase or not. In any case, the observed configuration by $r_j$ at $t_2$ is such that $r_i$ is distant from $p_j(t_0)$ by at most $d_0+d_0/3$. This implies that $ZoM^3_j(t_2)$ has a diameter of at most $(d_0+d_0/3)/3$, which implies that the destination point computed by $r_j$ in this cycle is distant from $p_j(t_0)$ by at most $d_0+d_0/3+d_0/9$. Repeating the argument, we get that $\forall t$, $ZoM^3_i(t)$ has a diameter $\leq \sum\limits_{i=1}^{\infty} d_0/3^i$. Reducing the formula, we obtain that $ZoM^3_i(t)$ is always $\leq d_0/2$, which implies that $ZoM^3_i(t) \subseteq ZoM^2_i(t)$.
\end{proof}

\paragraph*{\textbf{Ensuring Collision-freedom in Line Robocast Algorithms}}
To make LineRbcast1 and LineRbcast2 collision-free, it is expected
that any destination computed by a robot $r_i$ at 
$t$ be located within its $ZoM^3_i(t)$. The computation of
destinations is modified as follows: 
Let $dest_i(t)$ be the destination computed by a robot $r_i$ at time
$t$. Based on $dest_i(t)$, $r_i$ computes 
a new destination $dest^\prime_i(t)$ that ensures collision
avoidance. $dest^\prime_i(t)$ can be set to any point 
located in $[p_i(t_0), dest_i(t)] \cap ZoM^3_i(t)$. For example, we can
take $dest^\prime_i(t)$ to be equal to the point 
located in the line segment $[p_i(t_0), dest_i(t)]$ and distant from
$p_i(t_0)$ by a distance of $d_i(t)/2$ with $d_i(t)$ 
computed as explained in Definition \ref{defca:zom3}.

This modification of the destination computation method does not
impact algorithms correctness since it does not depend on the exact
value of computed destinations, but on the relationship between the
successive positions occupied by each robot. The algorithms remain
correct as long as robots keep the capability to freely change their
direction of movement and to move in both the positive and the
negative part of 
each such direction. This capability is not altered by the collision
avoidance scheme since the 
origin of the coordinate system of each robot - corresponding to its
original position - is strictly included in 
its zone of movement, be it defined by $ZoM^1$, $ZoM^2$ or $ZoM^3$.

\paragraph*{\textbf{Generalisation of the Protocols to $n$ Robots}}
As explained at the end of Section \ref{sec:lcs}, the generalisation of our algorithms to the case of $n$
robots has to deal with the issue of distinguishing robots from each others despite their anonymity. 
The solution we use is to instruct each robot to move in the close neighbourhood of its original position.
Thus, other robots can recognize it by comparing its current position with past ones.
For this solution to work, it is necessary that each robot always remains the closest one to all the positions it has previously occupied.
Formally speaking, we define the zone of movement $ZoM^4$ in a similar way as $ZoM^3$ except that
the diameter is this time equal to $d_i(t)/6$ (vs. $d_i(t) / 3$).
We now show that $ZoM^4$ provides the required properties. 
Let $r_i$ and $r_j$ be an arbitrary pair of robots and Let $d_{ij}$ denotes the distance between their initial positions.
It can easily shown, using the same arguments as the proof of Lemma~\ref{lemca:zom3}, that: 
\begin{enumerate}
\item Neither of the two robots moves away from its initial position by a distance greater than $d_{ij}/4$.
This implies that each robot remains always at a distance strictly smaller than $d_{ij}/2$ from  all the positions it has previously held.
\item The distance between $r_i$ (resp. $r_j$) and all the positions held by $r_j$ ($r_i$) is strictly greater than $d_{ij}/2$.
\end{enumerate}
Hence, $r_i$ can never be closer than $r_j$ to a position that was occupied by $r_j$, and vice versa.
This implies that it is always possible to recognize a robot by associating it with the position which is closest to it 
in some previously observed configuration.

\section{RoboCast Applications}
\label{sec:applications}

\subsection{Asynchronous Deterministic 2-Gathering} 

Given a set of $n$ robots with arbitrary initial locations and no agreement on a global coordinate system, $n$-Gathering requires that all robots eventually reach the same unknown beforehand location. $n$-Gathering was already solved when $n>2$ in both ATOM~\cite{SY99} and CORDA~\cite{CFPS03} oblivious models. The problem is impossible to solve for $n=2$ even in ATOM, except if robots are endowed with persistent memory~\cite{SY99}. In this section we present an algorithm that uses our RoboCast primitive to solve $2$-Gathering in the non-oblivious CORDA model.

A first "naive" solution is for each robot to robocast its abscissa and ordinate axes and to meet the other robot at the midpoint $m$ of their initial positions. RoboCasting the two axes is done using our Line RoboCast function described above in conjunction with the $ZoM^3-$based collision avoidance scheme.

A second possible solution is to refine Algorithm $\psi_{f-point(2)}$ of \cite{SY99,SY99b} by using our Line RoboCast function to "send" lines instead of the one used by the authors. The idea of this algorithm is that each robot which is activated for the first time translates and rotates its coordinate system such that the other robot is on its positive $y$-axis, and then it robocasts its (new) $x$-axis to the other robot using our Line Robocast function.
In \cite{SY99}, the authors give a method that allows each robot to 
compute the initial position of one's peer by comparing their two robocast $x$-axes defined above.
Then each robot moves toward the midpoint of their initial positions. 
Our Line RoboCast routine combined with the above idea achieves gathering in asynchronous systems 
within a bounded (\emph{vs.} finite in \cite{SY99}) number of movements
of robots and using only two 
(\emph{vs.} four) variables in their persistent memory.

\begin{thm}
There is an algorithm for solving deterministic gathering for two robots in non-oblivious asynchronous networks (CORDA).
\end{thm}

\subsection{Asynchronous Stigmergy}

Stigmergy \cite{DDPS09} is the ability of a group of robots that
communicate only through 
vision to exchange binary information.
Stigmergy comes to encode bits in the movements of robots.
Solving this problem becomes trivial when using our RoboCast primitive.
First, robots exchange their local coordinate system as explained in Section \ref{sec:lcs}.
Then, each robot that has a binary packet to transmit robocasts a
line to its peers whose angle with respect to its 
abscissa encodes the binary information.
Theoretically, as the precision of visual sensors is assumed to be
infinite, 
robots are able to observe the exact angle of this transmitted line, hence
 the size of exchanged messages may be infinite also.
However, in a more realistic environment in which sensor accuracy and calculations have a margin of error,
it is wiser to discretize the measuring space.
For this, we divide the space around the robot in several sectors 
such that all the points located in the same sector encode the same binary information (to tolerate errors of coding).
For instance, to send binary packets of 8 bits, each sector should have an angle equal to $u = 360^\circ / 2 ^ 8$.
Hence, when a robot moves through a line whose angle with respect to the abscissa is equal to $\alpha$, 
the corresponding binary information is equal to $\lfloor \alpha /n \rfloor$.
Thus, our solution works in asynchronous networks, uses a bounded
number of movements and also 
allows robots to send binary packets and not only single bits as in \cite{DDPS09}.
\section{Conclusion and Perspectives}
\label{sec:conclusion}

We presented a new communication primitive for robot networks, that can be used in fully asynchronous CORDA networks. Our scheme has the additional properties of being motion, memory, and computation efficient. We would like to raise some open questions:

\begin{enumerate}
\item The solution we presented for collision avoidance in CORDA can be used for protocols where robots remain in their initial vicinity during the whole protocol execution. A collision-avoidance scheme that could be used with all classes of protocol is a challenging issue.
\item Our protocol assumes that a constant number of positions is stored by each robot. Investigating the minimal number of stored positions for solving a particular problem would lead to interesting new insights about the computing power that can be gained by adding memory to robots.
\end{enumerate}

%\nocite{*}
%\begin{footnotesize}
\bibliographystyle{plain}
\bibliography{robocast}
%\end{footnotesize}

\clearpage
\appendix
\section*{Appendix}

\section{Correctness Analysis of LineRbcast1}
\label{sec:correctness_rbcast1}
In the following we prove that Algorithm \ref{alg-getview2} satisfies the specification 
of a RoboCast, namely validity and termination.
First we introduce some notations that will be further used in the proofs of the algorithm. 
For each variable $v$,
and each robot $r_i$, we denote $r_i.v(t)$ the value of the variable $v$ in the local memory of $r_i$ at time $t$.
When the time information can be derived from the context, we use simply $r_i.v$.

\paragraph*{Proof of the Validity property.} We start by the validity property. For this, we first prove a series of technical lemmas. 
The following two lemmas state the existence of a time instant at which both robots have reached $S_2$ 
and a following time instant at which at least one of them have reached $S_3$.

\begin{lem}
\label{lem2:state22}
Eventually, both robots reach state $S_2$.
\end{lem}

\begin{proof}
Thanks to the fairness assumption of the scheduler, every robot is activated infinitely often. 
The first time each robot is activated, 
it executes the lines (1.a, 1.b, 1.c) of Algorithm \ref{alg-getview2} and it reaches the state $S_2$.
\end{proof}

\begin{lem}
\label{lem2:state3}
Eventually, at least one robot reaches state $S_3$.
\end{lem}

\begin{proof}
Let $r_1$ and $r_2$ be two robots executing Algorithm \ref{alg-getview2}, 
and assume towards contradiction that neither of them reaches state $S_3$.
But according to Lemma \ref{lem2:state22}, they both eventually reach $S_2$.
Consider for each robot $r_i$ the cycle in which it reaches state $S_2$, 
and define $t_i$ to be the time of the end of the Look phase of this cycle.
Without loss of generality, assume $t_1 \leq t_2$ (the other case is symmetric).
Hence, the variable $r_1.pos_1$ describes the position of robot $r_2$ at $t_1$ expressed in the local coordinate system of robot $r_1$.
Let $t_3 > t_2$ be the time at which robot $r_2$ finishes its cycle that leads it to state $S_2$.
Between $t_2$ and $t_3$, robot $r_2$ performed a non null movement 
because it moved towards the point $(1.0)$ of its local coordinate system (line $1.b$ of the code).
Hence, the position of $r_2$ at $t_3$ is different from its position at $t_1 \leq t_2$ which was recorded in the variable $r_1.pos_1$.
By assumption, $r_2$ never reaches state $S_3$, so each time it is activated after $t_3$ 
it keeps executing the lines $2.a$ and $2.c$ of the code and 
never moves from its current position (reached at $t_3$).
By fairness, there is a time $t \geq t_3$ at which robot $r_1$ is activated again. At this time, it observes $r_2$ in a position different from $r_1.pos_1$.
This means that for $r_1$ the condition of line $2.a$ is false. Hence, $r_1$ executes the else block of the condition ($2.b.*$) and reaches state $S_3$ which contradicts the assumption and proves the lemma.
\end{proof}

\paragraph*{}
The next lemma expresses the fact that our algorithm exhibits some kind of synchrony 
in the sense that 
robots advance in the execution of the algorithm through the different states in unisson.
That is, neither of them surpasses the other by more than one stage (state).

\begin{lem}
\label{lem2:state43}
\label{lem2:state32}
If at some time robot $r_i$ is in state $S_j$ and robot $r_{1-i}$ is in state $S_k$ then $|j-k| \leq 1$.
\end{lem}

\begin{proof}
The proof of the lemma is divided into two parts:
\begin{itemize}

\item If robot $r_i$ is in state $S_3$ and robot $r_{1-i}$ is in state $S_j$ then $j \geq 2$.

\textit{proof: }
Since robot $r_i$ is in state $S_3$, it has necessarily executed the lines $(1.a \ldots 1.d)$ and $(2.b.*)$ of the code. 
Hence, the value of variable $r_i.pos_1$ is different from that of $r_i.pos_2$.
This means that $r_i$ has seen $r_{1-i}$ in at least two different positions which implies that $r_{1-i}$ has been activated at least once. 
Hence, $r_{1-i}$ has necessarily executed the lines $(1.a\ldots 1.c)$ and has reached $S_2$. 

\item If robot $r_i$ is in state $S_4$ and robot $r_{1-i}$ is in state $S_j$ then $j \geq 3$.

\textit{proof: }
For a robot $r_{1-i}$ to reach state $S_4$, 
it must execute the line $3.a$ of the code and detect that the other robot 
$r_{1-i}$ has changed its direction of movement (it moved toward the negative part of its x-axis). 
Thus, robot $r_{1-i}$ has necessarily executed lines ($2.b.2 \ldots 2.c$) of the code which means that $r_{1-i}$ is in state $S_3$. 
Before this, robot $r_{1-i}$ moved only in one direction, that is, in the positive direction of its x-axis.

\end{itemize}

%Since robot $r_i$ is in state $S_3$, it necessarily executed the lines $(1.a \ldots 1.d)$ and $(2.b.*)$ of the code, 
%hence the value of variable $r_i.pos_1$ is different from from that of $r_i.pos_2$.
%This means that $r_i$ has seen $r_{1-i}$ in a at least two different positions which implies that $r_{1-i}$ has been activated at least once. 
%Hence, $r_{1-i}$ has necessarily executed the steps $(1.a\ldots 1.c)$ and reached $S_2$. 
This proves the lemma.
\end{proof}

\paragraph*{•}
The following lemma states that each robot $r_i$ is guaranteed to be observed by its peer at least once 
in a position located in the positive part of its local x-axis.
The observed position is stored in $r_{1-i}.pos_1$. 
Expressed otherwise, this means that each robot "send" a position located in its positive x-axis to its peer. 
This property is important for proving validity 
which correspond to both robots eventually reaching $S_3$ (Lemmas \ref{lem2:state33} and \ref{lem2:validity}).
Indeed, since each robot $r_i$ is guaranteed that a point located in its local x-axis was received by its peer $r_{1-i}$, 
it suffices for $r_i$ to send its line to head toward the negative part of its x-axis 
and to stay there until it is observed by $r_{1-i}$.
That is, until a position located in its negative x-axis (and thus \emph{distinct} from $r_{1-i}.pos_1$) is received by $r_{1-i}$.
%Besides validy, the property of Lemma \ref{lem2:positivex} is also useful for computing chirality as will be shown next section.
%This property will also be useful in the computation of the chirality as will be explained in the following section.

\begin{lem}
\label{lem2:positivex}
For each robot $r_i$, the variable $r_i.pos_1$ describes a position located in the positive axis of the other robot $r_{1-i}$.
\end{lem}

\begin{proof}
%When robots are in state $S_1$ or $S_2$, they stay in their positive x-axis. 
The value of the variable $r_i.pos_1$ is assigned 
for the first time in line $1.a$ when $r_i$ is still in state $S_1$.  
At this time, according to Lemma \ref{lem2:state32}, $r_{1-i}$ is necessarily in state $S_1$ or $S_2$ 
(Otherwise, this would contradict Lemma \ref{lem2:state32} since we would have a time at which a robot ($r_{1-i}$) is in a state $S_j$ with $j \geq 3$ 
concurrently with another robot ($r_i$) that is in state $S_1$).
This means that the variable $r_i.pos_1$ describes a position held by robot $r_{1-i}$ while it was in state $S_1$ or $S_2$.
But according to the algorithm, when a robot is in state $S_1$ or $S_2$, it is still located in a position of its positive x-axis (or the origin).
Hence, the variable $r_i.pos_1$ describes a position 
located in the positive x-axis of robot $r_{1-i}$ which proves the lemma.
\end{proof}

\begin{lem}
\label{lem2:state33}
Eventually, both robots reach state $S_3$.
\end{lem}

\begin{proof}
According to Lemmas \ref{lem2:state3} and \ref{lem2:state32}, 
there is a time at which some robot, say $r_i$, reaches $S_3$ and the other one ($r_{1-i}$) is at a state $S_j$ with $j \geq 2$.
If $j \geq 3$ then the lemma holds and we are done. 
So we assume in what follows that $r_{1-i}$ is at $S_2$ 
and we prove that it eventually reaches $S_3$.
Assume for the sake of contradiction that this is not the case, that is, $r_{1-i}$ remains always stuck in $S_2$.
This implies, according to Lemma \ref{lem2:state43}, that $r_i$ remains also stuck in state $S_3$.
When $r_i$ is in state $S_3$, it keeps executing the line $3.b$ of the code until 
it reaches a position located in the negative part of its x-axis. 
Denote by $t_1$ the first time at which $r_i$ reaches its negative x-axis. 
Each time $r_i$ is activated after $t_1$, it executes the line $3.c$ of the code corresponding to a null movement and it never moves from its current position (located in the negative x-axis).
This is because we assumed that $r_i$ remains stuck in $S_3$ forever.
By fairness, there is a time $t_2 \geq t_1$ at which $r_{1-i}$ is activated.
This time, the condition of line $2.a$ does not hold for robot $r_{1-i}$ because the position returned by $observe(i)$ is located in the negative x-axis of $r_i$ and is different from $r_{1-i}.pos_1$ which is located in the positive part of the x-axis of $r_1$ (as stated in Lemma \ref{lem2:positivex}).
Hence, $r_{1-i}$ executes the part $2.b.*$ of the code and changes its status to $S_3$.
\end{proof}

\paragraph*{•}
Now we are ready to prove the validity property.

\begin{lem}
\label{lem2:validity}
Algorithm \ref{alg-getview2} satisfies the \textbf{validity} property of the Line Robocast Problem for two robots.
\end{lem}

\begin{proof}
Eventually both robots reach state $S_3$ according to Lemma \ref{lem2:state33}. 
Each robot $r_i$ in state $S_3$ has necessarily executed the blocks $1.*$ and $2.b.*$ of the algorithm
and thus delivered the line defined by the positions $r_i.pos_1$ and $r_i.pos_2$.
Now we prove that this line is well defined and that it does correspond to $l_{1-i}$, the line sent by $r_{1-i}$.
$r_i.pos_1$ and $r_i.pos_2$ are well defined since they are assigned a value in lines $1.a$ and $1.b.1$ respectively before $r_i$ delivers the line.
The assigned values correspond to two positions of $r_{1-i}$.
Moreover, by the condition of line $2.b.$ we have that these two positions are distinct.
It remains to prove that they belong to $l_{1-i}$.
% $r_1.pos_1$ and $r_i.pos_2$ are distinct.

%Hence 
%Hence, each robot in state $S_3$ has necessarily the values of its two variables $pos_1$ and $pos_2$ well defined and different from $\perp$. 
%By the condition of line $2.b.$ we have that $pos_1 \neq pos_2$ and both of them correspond to two \emph{distinct} positions of the other robot. 
The values of variables $r_i.pos_1$ and $r_i.pos_2$ are assigned when $r_i$ is in state $S_1$ and $S_2$ respectively.
Hence, according to Lemma \ref{lem2:state43}, when $r_i.pos_1$ and $r_i.pos_2$ are defined, 
$r_{1-i}$ did not yet reach $S_4$ and moved only through its x-axis.
This means that $r_i.pos_1$ and $r_i.pos_2$ correspond to two distinct positions of the x-axis of $r_{1-i}$.
Hence, $r_i$ delivered $l_{1-i}$.
Since both robots eventually reach $S_3$, both lines $l_i$ and $l_{1-i}$ are eventually delivered.
%Hence, $r_i$ can infer this x-axis (and consequently the line $l_{1-i}$) from $r_i.pos_1$ and $r_i.pos_2$ and deliver it at this stage.
%This holds for both robots which proves the lemma.
\end{proof}

%----------------------TERMINATION---------------------%
%----------------------TERMINATION---------------------%

\paragraph*{Proof of the Termination property.} 
Now we prove that the algorithm actually terminates. 
Before terminating, each robot $r_i$ must be sure that its peer $r_{1-i}$ has received its sent line, 
that is, $r_{1-i}$ has reached the state $S_3$.
As already explained, $r_i$ can infer the transition of $r_{1-i}$ to $S_3$ by detecting a change of its direction of movement.
Upon this, $r_i$ can go on to state $S_4$ and terminates safely. 
In the following two lemmas, we prove that at least one robot reaches $S_4$. 
To do this, we first prove in Lemma \ref{lem2:positive-pos_12} that at least one robot, say $r_{1-i}$, is observed by its peer in two distinct positions located in the positive part of its x-axis.
Later, when $r_{1-i}$ moves to its negative x-axis and $r_i$ observes it there, $r_i$ learns that $r_{1-i}$ changed its direction of movement which allows the transition of $r_i$ to state $S_4$. This is proved in Lemma \ref{lem2:state4}.

\begin{lem}
\label{lem2:positive-pos_12}
For at least one robot, say $r_i$, the two variables $r_i.pos_1$ and $r_i.pos_2$ describe two positions located in the positive x-axis of $r_{1-i}$ 
and such that $r_i.pos_2>r_i.pos_1$ with respect to the local coordinate system of $r_{1-i}$.
\end{lem}

\begin{proof}
Let $r_i$ be the first robot to enter state $S_3$. 
The other robot ($r_{1-i}$) is in state $S_2$ in accordance with Lemma \ref{lem2:state32}.
Hence, $r_{1-i}$ moved only through the positive direction of its x-axis, 
so the variables $r_i.pos_1$ and $r_i.pos_2$ correspond to two different positions in the positive x-axis of robot $r_{1-i}$ or in its origin. But since $r_i.pos_2$ was observed after $r_i.pos_1$ and $r_{1-i}$ moves in the positive direction of its x-axis, then $r_i.pos_2 > r_i.pos_1$ with respect to the local coordinate system of $r_{1-i}$.
\end{proof}

\begin{lem}
\label{lem2:state4}
Eventually, at least one robot reaches state $S_4$.
\end{lem}

\begin{proof}
%We have to prove that starting from some configuration in which 
%both robots are in state $S_3$, one of them, say $r_i$, leaves state $S_3$ and gets to state $S_4$.
We assume towards contradiction that no robot ever reach $S_4$.
But according to Lemma \ref{lem2:state33}, both robots eventually reach $S_3$.
Hence we consider a configuration in which both robots are in $S_3$ and we derive a contradiction by proving that at least one of them does reach $S_4$.
%This initial configuration is well defined in accordance with Lemma \ref{lem2:state33}.
Let $r_i$ be the robot induced by Lemma \ref{lem2:positive-pos_12}. 
The variables $r_i.pos_1$ and $r_i.pos_2$ of $r_i$ correspond to two different positions 
occupied by $r_{1-i}$ while it was on the positive part of its x-axis.
By assumption, $r_{1-i}$ eventually reaches state $S_3$. 
At the end of this cycle, 
$r_{1-i}$ is either located in a position of its negative x-axis or 
it keeps executing lines $3.b.*$ each time it is activated until it reaches such a position, let's call it $p$.
The next cycles it is activated, $r_{1-i}$ executes the line $3.c$ of the code because we assumed that $r_{1-i}$ never reaches $S_4$.
It results that $r_{1-i}$ never quits $p$.
Hence, $r_{1-i}$ is guaranteed to be eventually observed by $r_i$ in a position that is smaller than $r_i.pos_2$ with respect to the local coordinate system of $r_{1-i}$.
%After reaching a position of its negative x-axis, $r_{1-i}$ finishes by being observed by robot $r_i$ in a position that is smaller than $pos_2$. 
At this point, the condition of line $3.4$ becomes true for robot $r_i$, 
which executes the block of the code labelled by $3.4.*$ and sets is state to $S_4$.
\end{proof}

\begin{lem}
\label{lem2:state44}
Eventually, both robots reach state $S_4$.
\end{lem}

\begin{proof}
According to Lemma \ref{lem2:state4} at least one robot, say $r_i$, eventually reaches $S_4$.
When $r_i$ reaches $S_4$, $r_{1-i}$ is in a state $S_j$ with $ j\geq 3$ according to Lemma \ref{lem2:state43}.
If $j=4$ the lemma holds trivially, 
so we consider in the following a configuration in which $r_{1-i}$ is in state $S_3$
and we prove that it eventually joins $r_i$ in state $S_4$.
The variables $r_{1-i}.pos_1$ and $r_{1-i}.pos_2$ of $r_{1-i}$ describe two distinct positions located in the x-axis of robot $r_i$.
Let $p_i$ describes the position of $r_i$ at the end of the cycle in which it reaches $S_4$.
Once in state $S_4$, $r_i$ moves towards the point $myIntersect$ each time it is activated until it reaches it (lines $3.a.2$ and $4.a$ of the code).
$myIntersect$ is the point located at the intersection of $l_i$ and $nextl_i$ and its distance from $p_i$ is finite.
Since $r_i$ is guaranteed to move a minimal distance of $\delta_i$ at each cycle in which it is activated,
it reaches $myIntersect$ after a finite number of cycles.
The next cycle, $r_i$ chooses a destination located outside $l_i$ ($4.b.2$) and moves towards it before finishing the algorithm.
Let $t_i$ be the time of the end of the Move phase of this cycle 
and let $q_i$ be the position occupied by $r_i$ at $t_i$.
$q_i \notin l_i$ means that $q_i \notin line(r_{1-i}.pos_1, r_{1-i}.pos_2)$.
It follows that $r_{1-i}.pos_2 \notin  line(r_{1-i}.pos_1, q_i)$.
By fairness, there is a time $t>t_i$ at which $r_{1-i}$ is activated again, 
and at which it observes $r_i$ in the position $q_i$.
But we showed that $q_i$ is such that $r_{1-i}.pos_2 \notin  line(r_{1-i}.pos_1, q_i)$.
Hence the condition of line $3.a$ is true for robot $r_{1-i}$ in $t$ and it reaches state $S_4$ in this cycle.
\end{proof}

\begin{lem}
\label{lem2:termination}
Algorithm \ref{alg-getview2} satisfies the \textbf{termination} property of the Line Robocast Problem for two robots.
\end{lem}

\begin{proof}
Eventually, both robots reach $S_4$ as proved by Lemma \ref{lem2:state44}.
Let $r_i$ be a robot in state $S_4$ and 
let $p_i$ be its position at the end of the cycle in which it reaches $S_4$, 
Let $d_i$ be the distance between $p_i$ and $myIntersect_i$. 
Since the scheduler is fair and a robot is allowed to move in each cycle a minimal distance of $\sigma_i$ before it can be stopped by the scheduler, 
it follows that $r_i$ is guaranteed to cover the distance $d_i$ and to reach $myIntersect$ after at most $d_i/\sigma_i$ cycles.
The next cycle, $r_i$ moves outside $l_i$ and terminates.
\end{proof}

\begin{thm}
Algorithm \ref{alg-getview2} solves the Line Robocast Problem for two robots in unoblivious CORDA systems.
\end{thm}

\begin{proof}
Follows directly from Lemmas \ref{lem2:validity} and \ref{lem2:termination}.
\end{proof}

\newpage
\section{Generic RoboCast}
\label{sec:generic_cast}
In this section, we describe the RoboCast Algorithm for the general case of $n$ robots. Then we give its formal proof of correctness.
%Again, we ignore the collision avoidance issue in order to ease the presentation.
%However, our algorithm can be easily modified to become collision-free by incorporating to it the zone of movement technique already described in Section \ref{sec:collisionfree}.
%However, handling collisions Its incorporation in our algorithm is straightforward.

\label{sec:linemulticast-n}

\subsection{Description of the Algorithm}
The Line RoboCast algorithm for the general case of $n$ processes is a simple generalization of the algorithm for two robots.
The code of this algorithm is in Algorithm \ref{alg-getviewn}.                 
It consists in the following steps:
The first time a robot $r_i$ is activated in state $S_1$, it simply records the positions of all other robots in the array $pos_1[]$.
Then, it moves towards the point $(1,0)$ of its local coordinate system and goes to state $S_2$. 
When $r_i$ is in state $S_2$, each time it observes some robot $r_j$ 
in a position different from the one recorded in $pos_1[j]$, it stores it in $pos_2[j]$.
At this point, $r_i$ can infer the line sent by $r_j$ which passes through both $pos_1[j]$ and $pos_2[j]$. 
Hence, $r_i$ delivers $line(r_i, r_j)$ which corresponds to $l_j$.
$r_i$ does not move from its current position until it assigns a value to all the cells of $pos_2[]$ (apart from the one associated with itself which is meaningless).
That is, until it delivers all the lines sent by its peers.
Upon this, it transitions to $S_3$ and heads to the point $(-1, 0)$.
At state $S_3$, $r_i$ waits until it observes that all other robots changed the direction of their movement or moved outside their sent line.
Then, it moves towards a position located outside its current line $l_i$. In particular, it goes to a position located in $nextl_i$, the next line it will robocast.
Hence $r_i$ first passes by the intersection of $l_i$ and $nextl_i$. Then, it moves outside $l_i$ and terminates the algorithm.

\begin{algorithm}[htb]
\begin{scriptsize}

\caption{Line RoboCast \bf{LineRbcast1} for $n$ robots: Algorithm for robot $r_i$}          

%\caption{LineRbcast($l_i$, $nextl_i$) Function for $n$ robots: Algorithm for robot $r_i$}          
\label{alg-getviewn}                  
\begin{algorithmic}
\STATE \textbf{Variables}:\\

\STATE $state$: initially $S_1$.
\STATE $pos_1[1 \ldots n]$: initially $\perp$
\STATE $pos_2[1 \ldots n]$: initially $\perp$
\STATE $destination, intersection$: initially $\perp$
\STATE

\STATE   \textbf{Actions}:\\

\STATE \textbf{1. State [$S_1$]}: \emph{\%Robot $r_i$ starts the algorithm\%}
$\vspace{0.2cm}$
\STATE $\hspace{1cm}$a. \textbf{foreach } $1 \leq j \leq n$ \textbf{do }$pos_1[j] \leftarrow observe(j)$ \textbf{enddo}
\STATE $\hspace{1cm}$b. $destination \leftarrow (1,0)_i$
\STATE $\hspace{1cm}$c. $state \leftarrow S_2$
\STATE $\hspace{1cm}$d. Move to destination
\STATE

\STATE \textbf{2. State [$S_2$]}: \emph{\%$r_i$ knows at least one position of the lines of all other robots\%}
$\vspace{0.2cm}$
\STATE $\hspace{1cm}$a. \textbf{if} $\exists j \neq i$ \text{ s.t. } ($pos_2[j] = \bot$) \text{ and } ($pos_1[j] \neq observe(j)$) \textbf{then} 
\STATE $\hspace{2cm}$ 1. $pos_2[j] \leftarrow observe(j)$
\STATE $\hspace{2cm}$ 2. Deliver $(line(pos_1[j], pos_2[j])$ \textbf{endif}
\STATE $\hspace{1cm}$b. \textbf{if} $\exists j\neq i$ \text{ s.t. } ($pos_2[j] = \bot$) \textbf{then} $destination \leftarrow observe(i)$
%\STATE $\hspace{1cm}$b. \textbf{if} $\exists j$ \text{ such that } $pos_1[j] = observe(j)$ \textbf{then} $destination \leftarrow observe(i)$
\STATE $\hspace{1cm}$c. \textbf{else}
%\STATE $\hspace{2cm}$ 1. $pos_2 \leftarrow observe(1-i)$
\STATE $\hspace{2cm}$ 1. $destination \leftarrow (-1,0)_i$
\STATE $\hspace{2cm}$ 2. $state \leftarrow S_3$ \textbf{endif}
\STATE $\hspace{1cm}$d. Move to destination
\STATE

\STATE \textbf{3. State [$S_3$]}: \emph{\%$r_i$ knows the lines of all other robots\%}
$\vspace{0.2cm}$
\STATE $\hspace{1cm}$a. \textbf{if} $\forall j \neq i~ pos_2[j]$ is outside the line segment $[pos_1[j], observe(j)]$ 
\textbf{then}
%\STATE $\hspace{1cm}$a. \textbf{if} $\forall j, pos_2[j]$ is not between $pos_1[j]$ and $observe(j)$ in $line(pos_1[j], pos_2[j])$
\STATE $\hspace{2cm}$ 1. $intersection \leftarrow l_i \cap nextl_i$
\STATE $\hspace{2cm}$ 2. $destination \leftarrow intersection$
\STATE $\hspace{2cm}$ 3. $state \leftarrow S_4$
%\STATE $\hspace{2cm}$ 2. $destination \leftarrow (0,0_i)$
\STATE $\hspace{1cm}$b. \textbf{else if} ($observe(i) \geq (0,0)_i$) \textbf{then} $destination \leftarrow (-1,0)$
\STATE $\hspace{1cm}$c. \textbf{else} $destination \leftarrow observe(i)$ \textbf{endif endif}
\STATE $\hspace{1cm}$d. Move to destination
\STATE

\STATE \textbf{4. State [$S_4$]}: \emph{\%$r_i$ knows that all robots have learned its line $l_i$\%}
%\STATE $\hspace{1cm}$a. $intersectp \leftarrow intersection(l_i, nextl_i)$
\STATE $\hspace{1cm}$a. \textbf{if} ($observe(i) \neq intersection$) \textbf{then} $destination \leftarrow intersection$
\STATE $\hspace{1cm}$b. \textbf{else}
\STATE $\hspace{2cm}$1. $r_i$ rotates its coordinate system such that its x-axis and the origin match with
\STATE $\hspace{2cm}$	 $nextl_i$ and $intersection$ respectively.
\STATE $\hspace{2cm}$2. $destination \leftarrow (1,0)_i$
\STATE $\hspace{2cm}$3.  return \textbf{endif}
\STATE $\hspace{1cm}$c. Move to destination
\STATE

\remove{
\STATE
\STATE \textbf{4. State [$S_4$]}: \emph{\%Robot $r_i$ knows that all robots have learned its line the x-axis of all their peers\%}
\STATE $\hspace{1cm}$a. $destination \leftarrow (observe(i).x, y)_i$ with $y \neq 0$
\STATE $\hspace{1cm}$b. Move to destination
}

\end{algorithmic}

\end{scriptsize}
\end{algorithm}

\subsection{A Correctness Argument} We prove the correctness of our algorithm by proving that it satisfies the validity and termination property of the RoboCast Problem specification. The general idea of the proof is similar to that of the two robots algorithms
even if it is a little more involved.
%Only Lemma \ref{lemn:state32} is specific to the $n$ robots algorithm.

\paragraph*{Proof of the Validity property.}

\begin{lem}
\label{lemn:state22}
Eventually, all robots reach state $S_2$.
\end{lem}

\begin{proof}
Similar to the proof of Lemma \ref{lem2:state22}.
\end{proof}

\begin{lem}
\label{lemn:state3}
Eventually, at least one robot reaches state $S_3$.
\end{lem}

\begin{proof}
Let $\mathbb{R} = \{r_1, r_2, \ldots, r_n\}$ be a set of $n$ robots executing Algorithm \ref{alg-getviewn}.
We assume towards contradiction that neither of them ever reach $S_3$.
But according to Lemma \ref{lemn:state22}, all robots eventually reach state $S_2$.
Thus we proceed in the following way: we consider a configuration in which all robots are in $S_2$ and 
we prove that at least one of them eventually reaches $S_3$ which leads us to a contradiction.
%Hence, for each robot $r_i$ there is a cycle in which $r_i$ reaches state $S_2$.
Consider for each robot $r_i \in \mathbb{R}$ the cycle in which it reaches the state $S_2$,
and define $t_i$ and $t^\prime_i$ to be respectively the time of the end of the Look and the Move phases of this cycle.
Let $t_k$ be equal to $min \{t_1, t_2, \ldots, t_n\}$ and let $r_k$ be the corresponding robot. 
%That is, $r_k$ is the first robot to end a Look phase in the cycle in which it reaches state $S_2$.
That is, at $t_k$, robot $r_k$ finishes to execute a Look phase and at the end of this cycle it reaches state $S_2$.
This means that for robot $r_k$, the array $r_k.pos_1[]$ corresponds to the configuration of the network at time $t_k$.

Between $t_i \geq t_k$ and $t^\prime_i$ each robot executes complete Compute and Move phases.
The movement performed in this phase cannot be null because robots move from the point $(0,0)$
towards the point $(1,0)$ of their local coordinate system (line $1.b$ of the code).
Moreover, the scheduler cannot stop a robot before it reaches the point ($\delta_i$, 0).
Hence, the position of each robot $r_i$ at $t^\prime_i$ is different from its position at $t_i$.
But the position of $r_i$ at $t_i$ is equal to its position at $t_k$. 
Thus, the position of each robot at $t^\prime_i$ is different from its position at $t_k$ which is stored in $r_k.pos_1[i]$.
We have by assumption that no robot ever reaches $S_3$. 
So each time a robot $r_i$ is activated after $t^\prime_i$, it keeps executing the lines $2.b$ and $2.d$ of the code and never moves from 
its current position reached at $t^\prime_i$.
Define $t_{end}$ to be equal to $max\{t^\prime_1, \ldots, t^\prime_n\}$.
It follows that at $\forall t \geq t_{end}$ , the position of each robot $r_i$ at $t$ is different from $r_k.pos_1[i]$.
But by fairness, there is a time $t_a \geq  t_{end}$ at which $r_k$ is activated again. 
At this cycle, $r_k$ observes that each robot $r_i$ is located in a position different from $r_k.pos_1[i]$.
Consequently, if there exists a robot $r_i$ such that $r_k.pos_2[i]$ was equal to $\perp$ before this cycle, 
then $r_k$ assigns the current observed position of $r_i$ to $r_k.pos_2[i]$.
This implies that the condition of line $2.b$ is now false for $r_k$.
Hence $r_k$ executes the else block of the condition and reaches state $S_3$. This is the required contradiction that proves the lemma.
\end{proof}

\begin{lem}
\label{lemn:positivex}
For each robot $i$, for each robot $j$, if $r_i.pos_1[j] \neq \perp$, then $pos_1[j]$ describes a position that is necessarily located in the positive x-axis of robot $j$.
\end{lem}

\begin{proof}
The proof follows the same lines as that of Lemma \ref{lem2:positivex}.
\end{proof}

\begin{lem}
\label{lemn:unisson}
%\label{lem2:state32}
If at some time robot $r_i$ is in state $S_j$ and robot $r_i^\prime$ is in state $S_k$ then $|j-k| \leq 1$.
\end{lem}

\begin{proof}
The lemma can be proved by generalising the proof of Lemma \ref{lem2:state32} to the case of $n$ robots.
We divide the analysis into two subcases:

\begin{itemize}

\item If robot $r_i$ is in state $S_3$ and robot $r_i^\prime$ is in state $S_j$ then $j \geq 2$.

\textit{proof: }
If $r_i$ is in state $S_3$, this means that it observed all other robots in at least two distinct positions. 
This means that all other robots started a Move phase, which implies that they all finished a complete Compute phase in which they executed the lines $1.a \ldots 1.c$ of the code and reached $S_2$.

\item If robot $r_i$ is in state $S_4$ and robot $r_i^\prime$ is in state $S_j$ then $j \geq 3$.

\textit{proof: }
For a robot to reach $S_4$, 
it must detect a change of direction by \emph{all} other robots in the network which is captured by the condition of line $3.a$. 
We prove that this condition cannot be true unless all robots have reached $S_3$ and no robot in the network is still in state $S_2$.
Indeed, robots in state $S_1$ move in the positive direction of their x-axis and those in state $S_2$ does not move. 
So, a robot cannot change its direction before reaching state $S_3$.
This change of direction is reflected by the choice of point $(-1, 0)$ 
as a destination in line $2.c.1$ of the code before the transition to state $S_3$ in line $2.c.2$.
\end{itemize}
\end{proof}

\begin{cor}
\label{corn:unisson}
If at some time $t$, $\exists i, j$ such that robots $r_i$ and $r_j$ are respectively in state $S_k$ and $S_{k+1}$ at $t$ with $k \in \{1,2,3\}$,
then all the robots of the network are either in state $S_k$ or $S_{k+1}$ at $t$.
\end{cor}

\begin{proof}
Since $r_i$ is in state $S_k$, no robot in the network can be in a state $S_l$ with $l \geq k+2$ according to Lemma \ref{lemn:unisson}.
Similarly, the fact that $r_j$ is in state $S_{k+1}$ implies that no robot in the network can be in a state $S_l$ with $l\leq k-1$.
By the conjunction of the two facts, we obtain that all robots are either in state $S_k$ or $S_{k+1}$.
\end{proof}

\paragraph*{}
The following lemma proves the fact that if at time $t_a$, some robots of the network are in state $S_2$ and others are in state $S_3$, 
then at least one robot that is in state $S_2$ at $t_a$ eventually reaches $S_3$.

\begin{lem}
\label{lemn:state32}
Let $G_2(t), G_3(t)$ be the groups of robots that are respectively in state $S_2$ and $S_3$ at time $t$.
If at some time $t_a$~$\|G_2(t_a)\| > 0$ and $\|G_3(t_a)\| > 0$, then there exists a time $t \geq t_a$ at which $\|G_3(t)\| \geq \|G_3(t_a)\|+1 $.
\end{lem}

\begin{proof}
Since by assumption $\|G_2(t_a)\| > 0$ and $\|G_3(t_a)\| > 0$, 
it follows from Corollary \ref{corn:unisson} that $\mathbb{R} = G_2(t_a) \cup G_3(t_a)$.
We assume toward contradiction that $\forall t \geq t_a~ G_2(t) =G_2(t_a)=G_2$, 
that is, no robot that is in $S_2$ at $t_a$ ever reach $S_3$.
But by assumption we have $\|G_2\| > 0$. 
Hence $\forall t > t_a$ $\|G_2(t)\| = \|G_2\| >0$. 
This implies,  in accordance with Lemma \ref{lemn:unisson}, 
that no robot of the network can reach $S_4$ after $t_a$.
Consequently, $\forall t^\prime \geq t~ G_3(t^\prime) =G_3(t)=G_3 = \mathbb{R} \setminus G_2$.

\begin{itemize}
\item As discussed above, we have by assumption that all robots have reached $S_2$ at $t_a$.
This means that all robots have executed the line $1.a$ of the code at $t_a$. 
Consequently, at $t_a$, $\forall r_i \in \mathbb{R}$, $\forall r_j \in \mathbb{R}$ with $j \neq i$, $r_i.pos_1[j] \neq \perp$.
Moreover, according to Lemma \ref{lemn:positivex} all these positions stored in the arrays $pos_1[]$ describe positions located in the positive x-axis of the corresponding robots.

\item
By assumption we have that $\forall t > t_a$ $\|G_3(t)\| = \|G_3\|$. 
This means that no robot in $G_3$ ever reach $S_4$.
Hence robots of $G_3$ never execute the block $3.a.*$ of the code and 
they keep executing the line $3.b$ 
each time they are activated until they reach a position in their negative x-axes. 
Then, once a robot of $G3$ arrives to the negative part of its x-axis, it keeps executing the line $3.c$ of the code each time it is activated.
As we showed above, the positions stored in the different arrays 
$pos_1[]$ are different from $\bot$ and correspond to points located in the positive x-axes of the corresponding robots. 
Since robots of $G_3$ eventually get to positions in their negative x-axes and stay there, 
they eventually get observed by each robot in the network in a position 
different from the one that is stored in its local variable $pos_1[]$ which correspond to a positive x-axis position. 
Formally, there is a time $t_v > t$ at which $\forall r_i \in G_3, \forall r_j \in \mathbb{R}$ 
with $r_j \neq r_i$ $r_j.pos_2[i] \neq \perp$.

\item Let $r_1, \ldots, r_m$ be the robots of $G_2$. 
Consider for each robot $r_i \in G_2$ the cycle in which it reaches the state $S_2$,
and define $t_i$ and $t^\prime_i$ to be respectively the time of the end of the Look and the Move phase and of this cycle.
Let $t_k$ be equal to $min \{t_1, t_2, \ldots, t_m\}$ and let $r_k \in G_2$ be the corresponding robot. 
That is, at $t_k$, robot $r_k$ finishes to execute a Look phase and at the end of this cycle it reaches state $S_2$.
This means that for robot $r_k$, 
$r_k.pos_1[]$ describes the configuration of the network at time $t_k$. 
Following the lines of the proof of Lemma \ref{lemn:state3} we obtain that there exist a time 
$t_u>t$ at which $\forall r_j \in G_2 \setminus \{r_k\}, r_k.pos_2[j] \neq \bot$.
\end{itemize}

Now, let $t_x = max \{t_v, t_u\}$. 
From the discussion above it results that at time $t_x$, for robot $r_k \in G_2$ it holds that 
$\forall r_j \in G_3, r_k.pos_2[j] \neq \bot$ 
and $\forall r_j \in G_2 \setminus \{r_k\} , r_k.pos_2[j] \neq \bot$. 
Hence, at $t_x$, $\forall j \in \mathbb{R} \setminus \{r_k\}, r_k.pos_2[j] \neq \bot$. 
This means that the condition of line $2.a$ is false for $r_k$ at $t_x$, so $r_k$ executes the \emph{else} block of this condition when activated after $t_x$
and reaches state $S_3$ which
contradicts the assumption that $\forall t > t_a$ $\|G_2(t)\| = \|G_2\|$.
\end{proof}

\begin{lem}
\label{lemn:state33}
Eventually, all robots of the network reach state $S_3$
\end{lem}

\begin{proof}
Follows from Lemmas \ref{lemn:state22}, \ref{lemn:state3} and \ref{lemn:state32}.
\end{proof}

\begin{lem}
\label{lemn:validity}
Algorithm \ref{alg-getviewn} satisfies the \textbf{validity} property.
\end{lem}

\begin{proof}
The idea of the proof is similar to Lemma \ref{lem2:validity}.
According to Lemma \ref{lemn:state33}, all robots eventually reach state $S_3$.
Each robot that reach state $S_3$ has necessarily executed the block $1.*$ and the line $2.c$ of Algorithm \ref{alg-getviewn}. 
Hence, this robot has its two arrays $pos_1[]$ and $pos_2[]$ well defined and according to the way the elements of $pos_2[]$ are defined (refer to line $2.a$ of the code), we conclude that $\forall 1 \leq j \leq n$, $pos_2[j] \neq pos_1[j]$.
Moreover, since robots move only through their x-axes, $\forall j, pos_2[j] \text{ and } pos_1[j]$ correspond to two positions of the x-axis of robot $j$. 
Hence, each robot in state $S_3$ can infer the x-axes of its peers from $pos_1[]$ and $pos_2[]$ which proves the lemma.
\end{proof}

\remove{
\begin{proof}
Eventually both robots reach state $S_3$ according to Lemma \ref{lem2:state33}. 
Each robot in state $S_3$ has necessarily executed the blocks $1.*$ and $2.b.*$ of the algorithm. 
Hence, each robot in state $S_3$ has necessarily the values of its two variables $pos_1$ and $pos_2$ well defined and different from $\perp$. 
By the condition of line $2.b.$ we have that $pos_1 \neq pos_2$ and both of them correspond to two \emph{distinct} positions of the other robot. 
The values of variables $r_i.pos_1$ and $r_i.pos_2$ are defined when $r_i$ is in state $S_1$ and $S_2$ respectively.
Hence, according to Lemma \ref{lem2:state43}, when $r_i.pos_1$ and $r_i.pos_2$ are defined, 
$r_{1-i}$ did not yet reach $S_4$ and moved only through its x-axis.
This means that $r_i.pos_1$ and $r_i.pos_2$ correspond to two distinct positions of the x-axis of $r_{1-i}$.
Hence, $r_i$ can infer this x-axis (and consequently the line $l_{1-i}$) from $r_i.pos_1$ and $r_i.pos_2$ and deliver it at this stage.
This holds for both robots which proves the lemma.
\end{proof}
}
%---------------------------%
%------TERMINATION---------%
%---------------------------%

\paragraph*{Proof of the Termination property.}

\begin{lem}
\label{lemn:state4}
Eventually, at least one robot reaches state $S_4$
\end{lem}

\begin{proof}
We assume for the sake of contradiction that no robot ever reach $S_4$.
However, according to Lemma \ref{lemn:state33}, all robots eventually reach state $S_3$.
Hence we consider a configuration in which all robots are in state $S_3$ 
and we prove that at least one of them eventually reaches $S_4$ which leads us to a contradiction.
The idea of the proof is similar to that of Lemma \ref{lemn:state3}: 
we consider the first robot $r_k$ that executes a Look phase of a cycle leading it from $S_2$ to $S_3$.
Let $t_k$ be the time of the end of this Look phase.
Clearly, $\forall r_i \in \mathbb{R} \setminus \{r_k\}$, 
$r_k.pos_1[i]$ and $r_k.pos_2[i]$ describe two positions of $r_i$ located in its positive x-axis.
This is because these two positions were observed by $r_k$ 
before $r_i$ reaches $S_3$ and changes its direction of movement towards its negative x-axis.
Moreover, $r_k.pos_2[i] > r_k.pos_1[i]$ with respect to the local coordinate system of $r_i$ since $r_i$ 
was observed in $r_k.pos_1[i]$ and then in $r_k.pos_2[i]$ while it was moving along the positive direction of its x-axis.
The claim can be proved formally as in Lemma \ref{lem2:positive-pos_12}.
After $t_k$, all other robots of the network perform a transition from $S_2$ to $S_3$.
Then, they head towards the negative part of their local x-axes (lines $2.c.1$ and $3.b$ of the code) 
and stay there (line $3.c$) since they cannot reach $S_4$ by assumption.
Each robot $r_i$ that reaches the negative part of its x-axis is located in a position $p_i$ 
such that $r_k.pos_2[i]$ is outside the line segment [$r_k.pos_2[i]$, $p_i$].
Hence the condition of line $3.a$ eventually becomes true for robot $r_k$, and it reaches $S_4$ after executing the block $3.a.*$ of the code.
This is the required contradiction.
\end{proof}

%\begin{lem}
%\label{lemn:state43}
%If some robots of the network are in state $S_3$ and others are in state $S_4$, then at least one robot that is in state $S_3$ eventually reaches $S_4$.
%\end{lem}

\begin{lem}
\label{lemn:state44}
Eventually, all robots of the network reach $S_4$.
\end{lem}

\begin{proof}
The proof is similar to that of Lemma \ref{lem2:state44}.
The intuition behind it is as follows: we proved in Lemma \ref{lemn:state4} that at least one robot, say $r_i$, eventually reaches $S_4$.
After reaching $S_4$, and after a finite number of executed cycles, $r_i$ quits $l_i$ (line $4.b.2$).
When they observe $r_i$ outside $l_i$, the other robots transition to state $S_4$.
\end{proof}

\begin{lem}
\label{lemn:termination}
Algorithm \ref{alg-getviewn} satisfies the \textbf{termination} property.
\end{lem}

\begin{proof}
The proof is similar to that of Lemma \ref{lem2:termination}
\end{proof}

\begin{thm}
Algorithm \ref{alg-getview2} solves the Line Robocast Problem for $n$ robots in unoblivious CORDA systems.
\end{thm}

\begin{proof}
Follows directly from Lemmas \ref{lemn:validity} and \ref{lemn:termination}.
\end{proof}

\end{document}